\documentclass[aps,prl,twocolumn,superscriptaddress,groupedaddress, longbibliography]{revtex4}

\usepackage{graphicx}
\usepackage{dcolumn}
\usepackage{bm}
\usepackage{amssymb}   
\usepackage{physics}

\usepackage{amsthm}   
\usepackage{amssymb}   
\usepackage{physics}
\usepackage{times}
\usepackage{newtxmath}

\usepackage[colorlinks = true,
            linkcolor = orange,
            urlcolor  = orange,
            citecolor = orange,
            anchorcolor = orange, 
            allcolors = red]{hyperref}

\newtheorem{theorem}{Result}
\newtheorem{lemma}{Lemma}

\newcommand{\choi}[1]{\mathcal{J}^{\mathcal{#1}}}

\newcommand{\perm}[1]{\mathbb{P}_{#1}}

\newcommand{\Holevo}[1]{\chi{#1}} 

\newcommand{\freeS}[0]{\mathfrak{O}^{\rm{st}}} 
\newcommand{\staticRT}[0]{\mathcal{R}^{\rm{st}}}

\newcommand{\freeD}[0]{\mathfrak{O}^{\rm{dyn}}}
\newcommand{\dynamicRT}[0]{\mathcal{R}^{\rm{dyn}}}

\newcommand{\weylFreeD}[0]{\mathfrak{O}^{\rm{dyn}}_{\mathcal{W}}}

\newcommand{\depola}[1]{\mathcal{D}^{\rm{pol}}_{#1}} 
\newcommand{\dephase}[1]{\mathcal{D}^{\rm{ph}}_{#1}}

\begin{document}

\preprint{APS/123-QED}

\title{The Dynamical Resource Theory of Informational Non-Equilibrium Preservability}

\author{Benjamin Stratton}
\email{ben.stratton@bristol.ac.uk}
\affiliation{Quantum Engineering Centre for Doctoral Training, H. H. Wills Physics Laboratory and Department of Electrical \& Electronic Engineering, University of Bristol, BS8 1FD, UK}

\affiliation{H.H. Wills Physics Laboratory, University of Bristol, Tyndall Avenue, Bristol, BS8 1TL, UK}

\author{Chung-Yun Hsieh}

\affiliation{H.H. Wills Physics Laboratory, University of Bristol,
Tyndall Avenue, Bristol, BS8 1TL, UK}

\author{Paul Skrzypczyk}

\affiliation{H.H. Wills Physics Laboratory, University of Bristol,
Tyndall Avenue, Bristol, BS8 1TL, UK}

\affiliation{CIFAR Azrieli Global Scholars Program, CIFAR, Toronto Canada}

\date{\today}

\begin{abstract}
Information is instrumental in our understanding of thermodynamics. Their interplay has been studied through completely degenerate Hamiltonians whereby the informational contributions to thermodynamic transformations can be isolated. In this setting, all states other then the maximally mixed state are considered to be in informational non-equilibrium. An important yet still open question is: how to characterise the ability of quantum dynamics to maintain informational non-equilibrium? Here, the dynamical resource theory of informational non-equilibrium preservability is introduced to begin providing an answer to this question. A characterisation of the allowed operations is given for qubit channels and the n dimensional Weyl-covariant channels - a physically relevant subset of the general channels. An operational interpretation of a state discrimination game with Bell state measurements is given. Finally, an explicit link between a channels classical capacity and its ability to maintain informational non-equilibrium is made.
\end{abstract}

\maketitle


\section{Introduction}
Resources are precious. Their value arises from their limitation, incentivising them to be efficiently utilised and maintained. Formally, an object is considered to be a resource if it can be used by some agent to overcome the physical constraints of a system. A resource therefore allows an agent to obtain an otherwise impossible advantage. 

Within the quantum regime, the study of the limitations experienced by such an agent is the central aim of quantum resource theories. They provide a rigorous mathematical framework with which the ability of quantum objects to supply an operational advantage -- subject to some physical constraints -- can be compared \cite{Chitambar_2019}. Thus far, they have proved a fruitful way to approach problems concerning a number of different quantum phenomena, with resource theories of entanglement \cite{Horodecki_2009}, athermality \cite{Goold_2016, Horodecki2013, Lostaglio_2019}, and measurement, \cite{Skrzypczyk_2019} to name a few.

Initially, static resource theories (SRT) were the sole focus. In such SRTs, the primary objects in question are quantum states, and the resource is some property of these states. There is now an effort to expand resource theories beyond states to the dynamical regime \cite{gour2020dynamical, Gour_2020, ji2021convertibility, Hsieh_2020, Saxena2020, liu2019resource, Liu2020, Theurer2019, Rosset2018, Kim_2021}, with the aim to assess and compare the ability of different quantum operations to be resourceful, therefore providing some advantage. Within these dynamical resource theories (DRT), the resourceful objects are quantum operations with, some property of the quantum operation being considered a resource.

One such important property is the ability of operations to preserve the static resource present in the state upon which they act. These \textit{resource preservability theories}, introduced in \cite{Hsieh_2020}, are built upon SRTs and apply structure to their set of allowed operations. In particular, not all allowed operations are equal. Some operations will completely preserve the resource, while others will completely destroy it. How well an allowed operation preserves the static-resource can itself be considered a type of dynamic-resource, and this is the focus of such DRTs. 

Note that whilst intimately connected, the resource considered in the DRT and underlying SRT are fundamentally different and apply to different classes of objects - quantum channels and quantum states, respectively.
  
The desire for resource preservation arises naturally in any resource theory. Specifically, in the context of thermodynamics, it arises due to non-equilibrium states being considered a resource. For a given state, both its energy and the information an agent has about it will determine how resourceful it is \cite{Sparaciari_2020}. The contribution arising from information can be studied in isolation by considering trivial Hamiltonians, where all energy levels are degenerate. In this regime, any thermodynamic transformation that occurs must arise solely from an information theoretic origin, given there can be no change in energy. The maximally mixed state is the thermal state in this picture, and all other states are considered to be in \textit{informational non-equilibrium}. An understanding of this specific case can then be used to infer results about the general case, in which energetic considerations are necessary. 

The study of informational non-equilibrium has been succinctly formalised through the resource theory of informational non-equilibrium \cite{Horodecki_2003_Purity_RT, Gour_2015, Streltsov_2018}. {\color{black}Such SRTs often question the existence of transformations between states without any concern for the details of said transformations. This initially led to our understanding of thermodynamics having little to do with dynamics \cite{Lostaglio_2019}; however, insight has since been gained into the dynamical aspects of thermodynamic resource theories. For example, the largest set of feasible dynamics has been investigated \cite{Faist_2015}, and their consistency with the laws of thermodynamics \cite{Mischa_23} and practically relevant constraints assessed \cite{Lostaglio_2018, perry_2018}; their behaviour in the macroscopic limit has been considered \cite{Faist_2019}, and their analogue in the continuous variable regime studied \cite{Serafini_2020, Narasimhachar_2021}. Despite this, there has been no insight into how to describe the ability of a given dynamics to preserve informational non-equilibrium. Understanding what properties prevent thermalisation will not only enrich our foundational understanding of thermodynamics but also shed further light on the role of thermodynamics in information processing and transmission}. In addition, most quantum dynamics used in practical applications, such as quantum computation \cite{nielsen_chuang_2010} and quantum memories \cite{Heshami_2016}, benefit from this lack of thermalisation, which is a prevalent type of decoherence. An improved understanding of these dynamics could therefore further our control over these systems.

Within this work, we define and characterise the dynamical resource theory of informational non-equilibrium preservability to begin addressing this question. The informational non-equilibrium of a state will henceforth be referred to as the static-resource to distinguish it from the dynamic-resource. 

\section{Framework}
\subsection{The Resource Theory of Informational Non-Equilibrium}
The set of free states of a resource theory are those states that contain no resource. In the SRT of informational non-equilibrium, the only free state is the maximally mixed state, $\mathbb{I} / d_{S}$. Hence, the set of free states is $\{\Upsilon_{S} := \mathbb{I}_{S} / d_{S} \}$, where the subscript labels the subsystem and {\color{black}$d$ is the dimension.} 

{\color{black}All dynamics considered are quantum channels: completely positive trace-preserving linear maps. In a resource theory, dynamics is captured by the allowed operations.} These are resource non-generating physical manipulations that can be applied arbitrarily many times at no cost. Within the resource theory of informational non-equilibrium, these are taken to be the set of \textit{noisy operations}, $\freeS$, given by 
\begin{equation}
   \color{black} \mathcal{D}(\rho_{S}) = \textrm{Tr}_{E'} \Big[ U_{SE} \Big( \rho_{S} \otimes \Upsilon_{E} \Big) U_{SE}^{\dagger} \Big] \in \freeS, \label{noisyOperations}
\end{equation}
where $U_{SE}$ is a unitary operator on a joint system $SE$ \footnote{It will not be made explicit whether a state exists in the system or environment using a subscript S and E as it is mostly obvious from the context. Where it is not, the subscripts will be included.}. {\color{black}Physically, this means coupling the system $S$ with a bath $E$ in informational equilibrium, evolving the joint system under a global unitary, and then discarding some part of the joint system and environment, $E'$. }The resource theory of informational non-equilibrium will be referred to as $\staticRT$.

The existence of a noisy operation between two states $\rho$ and $\sigma$ is equivalent to the existence of a \textit{unital} channel, $\mathcal{N}$, such that $\sigma = \mathcal{N}(\rho)$ (a unital channel leaves the identity invariant)\cite{Gour_2015}. If $\rho$ and $\sigma$ are of equal dimension, this unital channel is a \textit{mixed unitary channel}.

\subsection{The Dynamical Resource Theory of Informational Non-equilibrium Preservability.}
Any SRT induces a DRT of static resource preservability \cite{Hsieh_2020}. The DRT of informational non-equilibrium is built upon $\staticRT$. It aims to apply structure to the set of noisy operations based on the channel's ability to preserve non-equilibrium. Formally, a noisy operation $\mathcal{N} \in \freeS$ is considered resourceful if it can output some non-equilibrium state, $\mathcal{N}(\rho) \neq \Upsilon_{S}$ for some $\rho$. The free objects of this resource theory are those channels that remove all of the resource from states upon which they act. In this case, since the free set of the SRT is a single state, the only free channel in the DRT is the state preparation channel of $\Upsilon_{S}$, $\Lambda(\cdot) := \textrm{Tr}(\cdot)\Upsilon_{S}$.

Allowed operations of a dynamical resource theory are \textit{super-channels} - linear mappings of quantum channels to quantum channels - that can only decrease the resource content of the channel \cite{Chiribella_2008}. Super-channels can be realised in the lab with a pre-processing and post-processing channel connected to a memory system. Using this form and the structure of resource non-generating super-channels presented in \cite{Hsieh_2020}, super-channels in the free set, $\freeD$, have the following form:
\begin{equation}
    \Pi(\mathcal{N}) = \sum_{\kappa} p_{\kappa} \mathcal{E}_{\kappa} \circ \mathcal{N} \circ \mathcal{P}_{\kappa} ~ \in ~ \freeD, \label{freeOperation}
\end{equation}
where shared randomness, described by the random variable $\kappa$ and a probability distribution $p_{\kappa}$, is permitted between the pre-processing and post-processing noisy operations $\mathcal{P}_{\kappa}, \mathcal{E}_{\kappa}$. It can be verified that these super-channels map noisy operations to noisy operations, with $\Lambda$ only being mapped to itself. Further details on the allowed super-channels are given in Supplementary Material A \cite{suppMat}. 

The physical reasoning behind this set of super-channels being allowed is that all parts come from the allowed set of $\staticRT$. Given noisy operations only increase informational non-equilibrium, the output states of $\Pi (\mathcal{N})$ cannot contain any more static-resource than the output states of $\mathcal{N}$. Super-channels in $\freeD$ are therefore dynamic-resource non-increasing, as $\Pi (\mathcal{N})$ cannot preserve more {\color{black}informational non-equilibrium} present in any state upon which it acts than $\mathcal{N}$. The dynamical resource theory of informational non-equilibrium preservability will be referred to as $\dynamicRT$. 

All resource theories aim to quantify how resourceful a given object is. Since allowed operations cannot generate any resource, we can understand this by studying which objects can be converted into each other. In $\staticRT$, this characterisation takes the following form: given two quantum states $\rho$ and $\sigma$, does there exist a noisy operation mapping $\rho \rightarrow \sigma$? Answering this establishes a pre-order on the set of all states based on the existence of a noisy operation between them. This pre-order then allows the state's resource content to be compared. For the SRT, this convertibility question has been answered and is simple, being captured entirely by the mathematical notion of majorisation \cite{Gour_2015}. Answering it for $\dynamicRT$ is the aim of this work. This will establish a pre-order on the set of channels $\freeS$ based on the existence of an allowed super-channel between them. Given these are resource non-increasing, if there exists an allowed super-channel mapping $\mathcal{N} \rightarrow \mathcal{M}$, where $\mathcal{M,N} \in \freeS$, it will be known that $\mathcal{M}(\rho)$ will be closer to $\Upsilon_{S}$ then $\mathcal{N}(\rho)$ for all $\rho$. In other words, $\mathcal{M}$ will not preserve {\color{black}informational non-equilibrium} better than $\mathcal{N}$.  

\subsection{Characterising Allowed Super-channels by a Choi-state
Representation} 
To answer the convertibility question, the Choi-Jamiołkowski isomorphism will be employed. This is a linear mapping between a quantum channel $\mathcal{N}$ and a bipartite quantum state $\choi{N}$, given by 
\begin{equation}
    \choi{N} = (\mathcal{I} \otimes \mathcal{N}) ( \ketbra{\Phi} ), 
\end{equation}
where $\ket{\Phi} := \frac{1}{\sqrt{d}} \sum_{i} \ket{ii}$. Specifically, as proven in Supplementary Material A \cite{suppMat}, if there exists an allowed operation $\Pi \in \freeD$, so that $\mathcal{M} = \Pi(\mathcal{N})$ in the form of Eq.~\eqref{freeOperation}, for two noisy operations $\mathcal{M}$ and $\mathcal{N}$, then the corresponding Choi-states are related as
\begin{equation}
    \begin{split}
        \choi{M} &= \sum_{\kappa} p_{\kappa} (\mathcal{P}_{\kappa} \otimes \mathcal{E}_{\kappa} ) \big( \choi{N} \big). \label{Choi State allowed operations}
    \end{split}
\end{equation} 
{\color{black} In what follows, the dynamics of a quantum system of a constant finite dimension are considered. In this case, how far a state is from informational non-equilibrium is equivalent to asking how pure that state is \cite{Gour_2015}, see the appendix for more details. As mentioned above, under this restriction $\mathfrak{O}^{\textrm{st}}=\mathfrak{U}_{M}$, where $\mathfrak{U}_{M}$ is the set of mixed unitary channels.}

\section{Characterising Allowed Operations for Qubit Systems}
Initially, qubit channels are focused on, such that our resource theory concerns qubit mixed unitary channels.

A characterisation of $\dynamicRT$ looks to answer the following question: \textit{Given two channels $\mathcal{M,N}\in\freeS$, does there exist an allowed operation $\Pi \in \freeD$ such that $\mathcal{M} = \Pi(\mathcal{N})$.} Using the form of allowed operations given in Eq.~(\ref{Choi State allowed operations}), the above question equivalently asks if there exists some convex combination of channels  $\mathcal{P}_{\kappa}, \mathcal{E}_{\kappa} \in \freeS$ that act on the subsystems of $\choi{N}$ respectively to map $\choi{N} \rightarrow \choi{M}$. It is shown through the following result that this depends only upon the eigenvalues of the corresponding Choi-states. In what follows, $\bm{\mu}$ and $\bm{\lambda}$ are vectors of eigenvalues of the Choi states $\choi{M}$ and $\choi{N}$, respectively. 
\begin{theorem}\label{Result1}
For every given $\mathcal{M,N} \in \freeS$, the following statements are equivalent:
\begin{enumerate}
\item $\mathcal{M}=\Pi(\mathcal{N)}$ for some allowed operation $\Pi \in \freeD$.
\item $\bm{\mu} = \mathbb{D} \bm{\lambda}$, where $\mathbb{D} = \sum_{n,m} p_{nm} ~ \sigma_{x}^{n} \otimes \sigma_{x}^{m}$ with $p_{nm}$ the elements of a probability vector and $\sigma_{x}$ the Pauli X.
\end{enumerate}
\end{theorem}
Result~\ref{Result1} states that an allowed operation of $\dynamicRT$ exists between $\mathcal{N}$ and $\mathcal{M}$ if and only if the vector of eigenvalues, $\bm{\mu}$, of the Choi-state of the channel $\mathcal{M}$ is in the convex hull of local permutations of the vector of eigenvalues, $\bm{\lambda}$, of the Choi-state of the channel $\mathcal{N}$. This is reminiscent of the analogous result for $\staticRT$ \cite{Gour_2015} since majorisation is equivalent to the existence of a doubly stochastic matrix mapping the eigenvalues of $\rho$ and $\sigma$ \cite{bapat_raghavan_1997}. Here, in contrast, we don't have all doubly stochastic matrices but only a subset, those involving \emph{local} permutations. A sketch of the proof of Result~\ref{Result1} is provided in the appendices below, and a complete derivation can be found in Supplementary Material B \cite{suppMat}. 

Mathematically, the question of the existence of an allowed operation between two channels can now be simplified to a convex optimisation problem that can be easily solved on conventional computing hardware. Physically, this allows a comparison of the ability of the allowed operations of $\staticRT$ to preserve purity.    

As an application of Result \ref{Result1}, we consider the existence of an allowed operation between the depolarising, $\depola{s}$, and dephasing, $\dephase{q}$, channels, parameterised by $s$ and $q$ respectively (if $s,q = 1$ they are the identity channel; if $s,q = 0$ they are the completely depolarising/dephasing channel). It is found that there exists an allowed operation of $\dynamicRT$ mapping $\depola{s} \rightarrow \depola{s'}$ if and only if $s' \leq s$, and similarly, $\dephase{q} \rightarrow \dephase{q'}$ if and only if $q' \leq q$. This confirms the obvious conclusion that increasing the probability of depolarising or dephasing decreases the amount of preserved purity. In addition, no allowed operation exists that maps $\depola{s} \rightarrow \dephase{q}$ except in the trivial case ($s=1$), as expected. However, there does exist an allowed operation mapping $\dephase{q}\rightarrow \depola{s}$ if and only if $s \leq q/(2-q)$. See Supplementary Material B \cite{suppMat} for the details. 
\begin{figure}
    \centering
    \hspace*{0cm}\includegraphics[scale=0.32]{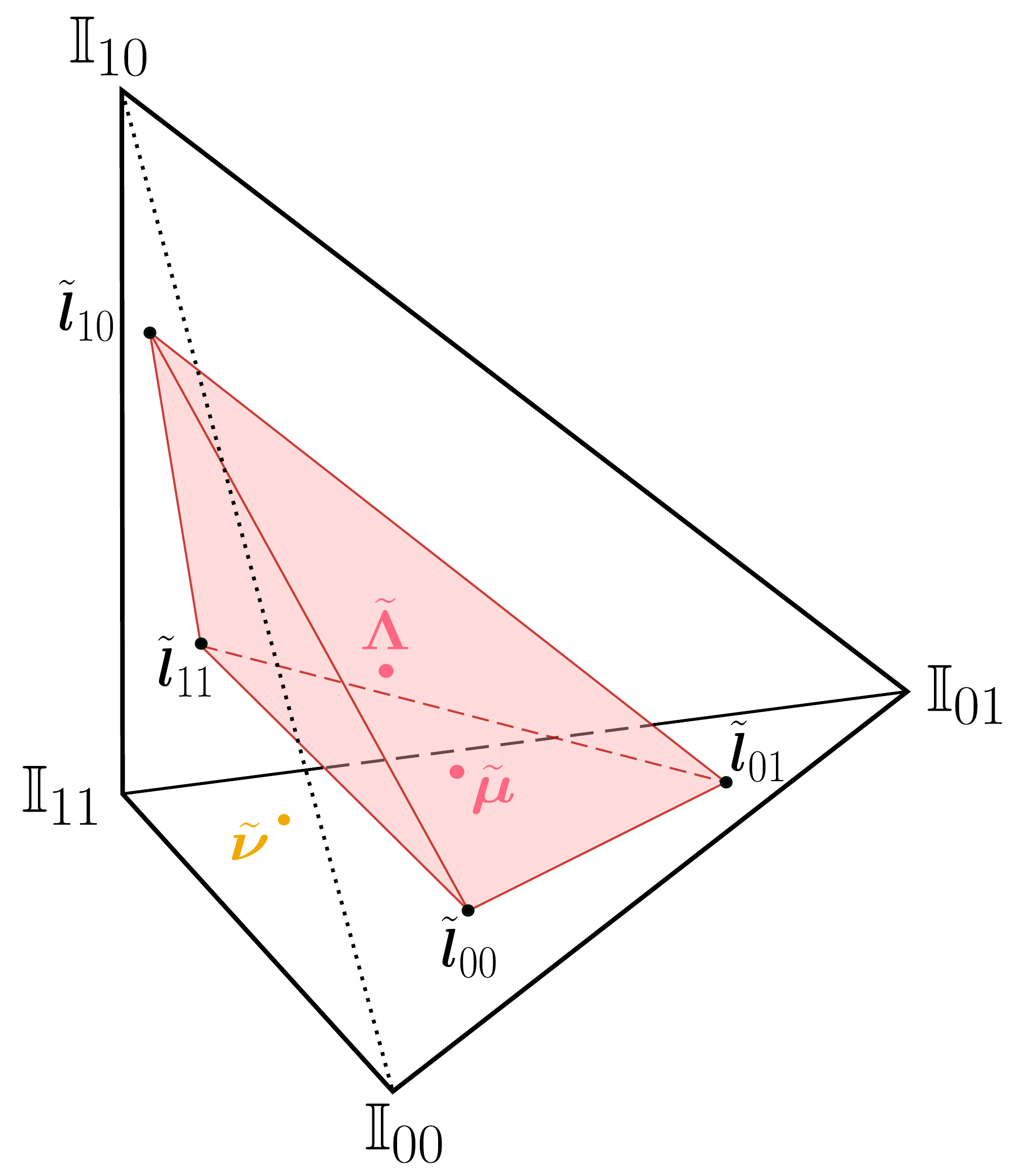}
    \caption{A simplex in $\mathbb{R}^{3}$ where each point represents the Choi-state of a qubit unital channel that is diagonal in the Bell basis. $\mathbb{I}_{ij}$, the extreme points of the outer simplex, are the identity channel and channels that are equivalent to the identity channel up to a Pauli operator. The set of vectors $\Tilde{\mathfrak{L}}$ are plotted as points, $\Tilde{\bm{l}}_{ij}$, where each point represents a local permutation of $\bm{\lambda}$. The convex hull of these vectors is shaded red with the free channel, given by $\Tilde{\bm{\Lambda}}$, at the centre. $\Tilde{\bm{\mu}}$ sits within this convex hull, and hence, an allowed operation exists between the channel $\bm{\lambda}$ and $\bm{\mu}$. $\Tilde{\bm{\nu}}$ sits outside, and hence, no allowed operation exists between the channel $\bm{\lambda}$ and $\bm{\nu}$.
    }
    \label{fig:1}
\end{figure}

\subsection{A Complete Set of Monotones} 
A monotone of a resource theory is a function $M(\cdot)$ such that $M(\mathcal{N}(\rho)) \leq M(\rho)$ where $\mathcal{N} \in \freeS$. Although, it may be the case that $M(\sigma) \leq M(\rho)$ even if no allowed operation exists between $\rho$ and $\sigma$. \textit{A complete set of monotones}, $\{M_{i}\}$, does, however, provide an alternative to the convertibility question, with an allowed operation existing such that $\rho \rightarrow \sigma$ if and only if $M_{i}(\sigma) \leq M_{i}(\rho) ~\forall ~ i$. The characterisation given in Result~\ref{Result1} admits a geometrical interpretation of an allowed operation through which a complete set of monotones can be found.

The set of local permutations of $\bm{\lambda}$ is $ \mathfrak{L} = \{ \perm{i} \bm{\lambda} : \forall ~ \perm{i} \in \mathfrak{P} \}$, where $\mathfrak{P}$ is the set of local qubit permutations. Result~\ref{Result1} can be rephrased, stating that a matrix $\mathbb{D}$ exists, such that $\bm{\mu} = \mathbb{D}\bm{\lambda}$, if $\bm{\mu}$ is in the convex hull of vectors in $ \mathfrak{L}$. Each element of $\mathfrak{L}$ lives in a three-dimensional subspace of $\mathbb{R}^4$ due to normalisation. The vectors $\mathfrak{L}$ within this lower dimensional subspace are given by $\Tilde{\mathfrak{L}}$. Geometrically, the elements of $\Tilde{\mathfrak{L}}$ form a simplex, as depicted in Fig.~\ref{fig:1}. Each point represents a qubit unital channel diagonal in the Bell basis. The points of the outer simplex are the identity channel and channels that are equivalent to the identity up to a Pauli operator. From these channels, it is possible to reach all other channels under allowed operations, as their convex hull is the whole space. The centre of the simplex, $\Tilde{\bm{\Lambda}}$, is the free state preparation channel, $\Lambda$. This channel can be reached by any channel under allowed operations.

Using this geometrical representation, a complete set of monotones can readily be identified. In what follows, we consider all vectors in $\Tilde{\mathfrak{L}}$ to be linearly independent, which is the general case \footnote{When considering any $\epsilon$ level of noise in the description of a given channel, the vectors in $\Tilde{\mathfrak{L}}$ will always be linearly independent.}. For completeness, all cases are considered in Supplementary Material C \cite{suppMat}.  
\begin{theorem}
    \textit{There exists a set of $d_{S}^{2}=4$ linear inequalities that are a complete set of monotones of $\dynamicRT$ if $\Tilde{\mathfrak{L}}$ is a set of linearly independent vectors. } \label{result2}
\end{theorem}
With the vectors in $\Tilde{\mathfrak{L}}$ being linearly independent, their convex hull forms a tetrahedron in $\mathbb{R}^{3}$. If the three-dimensional representation of the channel $\bm{\mu}$, given by $\Tilde{\bm{\mu}}$, lies inside this tetrahedron, it is then in the convex hull of the vectors in $\Tilde{{\mathfrak{L}}}$. This is equivalent to $\bm{\mu}$ being in the convex hull of the vectors in $\mathfrak{L}$ and hence an allowed operation $\bm{\lambda} \rightarrow \bm{\mu}$ existing.  

Each monotone then checks on what side of the planes that encompass each of the four faces of the tetrahedron $\Tilde{\bm{\mu}}$ lies on. If for all four planes $\Tilde{\bm{\mu}}$ lies on the same side as $\Tilde{\bm{\Lambda}}$, it must be inside the tetrahedron, and an allowed operation exists. 

\section{Characterising Allowed Operations for Arbitrary Finite Dimensional Systems}
Result~\ref{Result1} is extended to higher dimensional systems, such that $d_{S}=n$, in the important specialised case that $\freeS$ is the set of \textit{Weyl-covariant channels}, $\mathfrak{U}_{\mathcal{W}}$, where $\mathfrak{U}_{\mathcal{W}} \subsetneq ~ \mathfrak{U}_{M}$.  A channel $\mathcal{N}$ is Weyl-covariant if $\mathcal{N}(W_{ij}\rho W^{\dagger}_{ij}) = ~ W_{ij}\mathcal{N}(\rho)W^{\dagger}_{ij}~\forall~ i,j, \rho$ where $W_{ij}$ are the discrete Weyl operators \cite{watrous_2018}. These are a physically relevant set of channels given that they contain the completely depolarising, completely dephasing and partial thermalisation channels. More detail on these channels can be found in \cite{watrous_2018} and Supplementary Material D \cite{suppMat}. Our allowed operations now become a pre-processing and post-processing Weyl-covariant channel with shared randomness between them, given by the set $\weylFreeD \subsetneq \freeD$. The DRT considered for $d_{S} = n$ will be referred to as $\dynamicRT_{n}$. In this picture, we get the following result:
\begin{theorem}\label{result3}
For every given $\mathcal{M,N} \in \freeS=\mathfrak{U}_{\mathcal{W}}$, the following statements are equivalent:
\begin{enumerate}
\item $\mathcal{M}=\Pi(\mathcal{N)}$ for some allowed operation $\Pi \in \weylFreeD$.
\item $\bm{\mu} = \mathbb{D}' \bm{\lambda}$, where $\mathbb{D}'$ is a convex combination of local \textbf{cyclic} permutation matrices.  
\end{enumerate}
\end{theorem}
A sketch of the proof of Result~\ref{result3} can be found in the Appendices below, and a complete derivation can be found in Supplementary Material D \cite{suppMat}.

Following the same logic as in the qubit case, a complete set of monotones can again be found.
\begin{theorem} \label{result4}
    \textit{There exists a set of $d_{S}^{2}$ inequalities that are a complete set of monotones of $\dynamicRT_{n}$ if $\Tilde{\mathfrak{L}}$ is a set of linearly independent vectors.}
\end{theorem}

\section{Characterising Allowed Operations via State Discrimination} 
The complete set of monotones of $\dynamicRT$ and $\dynamicRT_{n}$ can be used to create an operational interpretation of the pre-order established on the set of noisy operations.   
\begin{theorem} \label{result5}
Let $\mathcal{N}\in\freeS$ be given.
Then there exists a set of $d_S^2$ many bipartite states $\{\rho_{ij}\}$ making the following two statements equivalent:
\begin{enumerate}
\item $\mathcal{M}=\Pi(\mathcal{N})$ for some $\Pi \in \freeD$.
\item For every $i,j$ we have
\begin{equation}
\bra{\Phi_{ij}} ( \mathcal{I} \otimes \mathcal{N}) (\rho_{ij}) \ket{\Phi_{ij}} \geq \bra{\Phi_{ij}} (\mathcal{I} \otimes \mathcal{M}) (\rho_{ij}) \ket{\Phi_{ij}}.
\end{equation}
\end{enumerate}
Here, $\ket{\Phi_{ij}}$ is a state from the Bell basis and $\Tilde{\mathfrak{L}}$ must be a set of linearly independent vectors.
\end{theorem} 
This can be interpreted as a state discrimination task: firstly, a referee distributes half of the bipartite state $\rho_{ij}$ to Alice and half to Bob. Alice then sends her half to Bob via the channel $\mathcal{N}$ or $\mathcal{M}$, and Bob makes a Bell-state measurement. Bob succeeds if he gets the measurement outcome associated with $\ket{\Phi_{ij}}$ for the state $\rho_{ij}$. If for all states in the set Bob's success probability is at least as high when Alice applies $\mathcal{N}$ as when Alice applies $\mathcal{M}$, an allowed operation exists between $\mathcal{N}$ and $\mathcal{M}$. Here, the individual probabilities of successful discrimination need to be considered, not just the average success probability, as is common in other state discrimination tasks \cite{Bae_2015}. 

This result allows the pre-order to be experimentally investigated. Each state, $\rho_{ij}$, is constructed in a similar method to the complete set of monotones. Interestingly, there is some freedom in how this operational interpretation is formed. This is detailed with a derivation in Supplementary Material E \cite{suppMat}.  

\section{Classical Capacity Quantifies Informational Non-Equilibrium Preservability} 
In addition to finding complete sets of monotones, it is important to find physically motivated monotones. The \textit{Holevo capacity} - a lower bound on the classical capacity of a quantum channel \cite{watrous_2018, PhysRevLett.83.3081, PhysRevA.55.1613, 1377491, 8242350, holevo2002remarks} - is one such physically relevant monotone of both $\dynamicRT$ and $\dynamicRT_{n}$. 

If there exists a $\Pi \in \freeD$ such that $\mathcal{M}=\Pi(\mathcal{N})$, where $\mathcal{M,N} \in \freeS$, then $\Holevo(\mathcal{N}) \geq \Holevo(\mathcal{M})$ where $\Holevo(\cdot)$ is the Holevo capacity. See Supplementary Material E \cite{suppMat} for a full derivation. Therefore, $\dynamicRT$ can be employed to compare the classical communication abilities of unital quantum channels in settings where non-equilibrium can only increase. {\color{black} See the appendices below for details of how this relates to the dynamical resource theory of communication \cite{Takagi_2020}, and see Supplementary Material G \cite{suppMat} for more details on monotones.}

\section{Conclusion} 
Here, the dynamical resource theory of informational non-equilibrium preservability has been characterised for both qubit mixed unitary channels and Weyl-covariant channels. Through the presented framework, the ability of two allowed operations of $\staticRT$ to preserve the purity of input states can be compared for all states in the space of states. This characterisation acts as a proof of principle for resource preservability theories and suggests that others could be developed. The clear and immediate next step would be to extend these results to the resource theory of thermodynamics. An understanding of the ability of \textit{thermal operations} \cite{Horodecki2013} to preserve non-equilibrium could further our knowledge of the dynamical aspects of quantum thermodynamics. Such understanding could also be utilised, for example, in the study of efficient thermal processes. {\color{black}Moreover, a characterisation of the DRT of information non-equilibrium for channels with differing input and output dimension would be another interesting extension of the above results.} Additionally, resource preservability theories could also be built upon other successful SRTs of a more pragmatic nature. For instance, the set of local operations and classical communication could be studied and their ability to preserve entanglement quantified.  

Moreover, a link between the ability of a channel to preserve information non-equilibrium and classical capacity - a natural measure for classical communication - has been made. Further effort should be made to discover other physically relevant monotones of $\dynamicRT$ in the hope of finding additional areas in which it is applicable.
\begin{acknowledgments}
\textbf{\textit{Acknowledgments}}. B.S. acknowledges support from UK EPSRC (EP/SO23607/1). P.S. and C.-Y.H. acknowledge support from a Royal Society URF (NFQI). C.-Y.H. also acknowledges support from the ERC Advanced Grant (FLQuant). P.S. is a CIFAR Azrieli Global Scholar in the Quantum Information Science Programme. Fig.~\ref{fig:1} was made with the help of \cite{gg2}.
\end{acknowledgments}

{\color{black} \textit{\textbf{Appendix on Informational Non-equilibrium and Purity}}
Whilst similar, the resource theory of informational non-equilibrium and the resource theory of purity differ in what exactly they consider a resource. Two states that are identical on their support but embedded in different dimensional Hilbert spaces will have the same amount of resource in the resource theory of purity but a differing amount of resource in the resource theory of informational non-equilibrium. In a Hilbert space of fixed dimension, these resource theories become alternative interpretations of the same physics. Where context permits, informational non-equilibrium will be referred to as purity in this work. This resource theory could, therefore, also be referred to as the dynamical resource theory of purity preservability. }

\textit{\textbf{Appendix on Results 1}}
Here we provide a sketch of the proof of Result~\ref{Result1}. A complete derivation can be found in Supplementary Material B \cite{suppMat}. 
\begin{proof}
Firstly, we prove that statement 1 implies statement 2. We simplify the form of the allowed operations through the following lemma:
\begin{lemma} \label{lemma1}
\textit{The Choi-states of qubit unital channels can always be digonalised in the Bell basis (maximally entangled basis) under local unitaries.}
\end{lemma}
Given all unitary channels are in $\freeS$, there always exists an allowed operation that can diagonalise $\choi{N}$. The action of all qubit unital channels can, therefore, be captured through only the eigenvalues of their Choi-state. Equation~\eqref{Choi State allowed operations} can now be rephrased as $\bm{\mu} = \mathbb{B} \bm{\lambda}$, where $\bm{\lambda}, \bm{\mu} \in ~ \mathbb{R}^{4}$ are vectors of eigenvalues of $\choi{N}$, $\choi{M}$ respectively, and $\mathbb{B}$ is in the convex hull of matrices with elements $B_{nm, kl} = ~ \abs{ \bra{\Phi_{kl}} U \otimes V \ket{\Phi_{nm}} }^{2}$ for some general unitaries $U, V$. The Bell basis states are given by $\ket{\Phi_{ij}} = (\mathbb{I} \otimes W_{ij})\ket{\Phi_{00}}$ where $W_{ij}$ is a discrete Weyl operator and $\ket{\Phi_{00}} := \ket{\Phi}$ (see \cite{watrous_2018} and Supplementary Material A \cite{suppMat} for details on Weyl operators). Matrices of this type are a subset of the unistochastic matrices \cite{Bengtsson_2005} - we coin this subset \textit{the product-Bell unistochastic matrices}. The following lemma is now employed: 
\begin{lemma} \label{lemma2}
\textit{Product-Bell unistochastic matrices are a convex combination of local permutation matrices.}
\end{lemma} 
The proof of this direction is then completed by noting that a convex combination of product-Bell unistochastic matrices remains a convex combination of local permutation matrices.

To show statement 1 implies statement 2, consider the channels $\widetilde{\mathcal{N}},\widetilde{\mathcal{M}}$ whose Choi-states are diagonalised in the given Bell basis with the same eigenvalues as $\mathcal{J}^\mathcal{N},\mathcal{J}^\mathcal{M}$.
This is guaranteed by Lemma~\ref{lemma1}. When statement 2 holds, one can write $\widetilde{\mathcal{M}} = \alpha \widetilde{N} + (1-\alpha) \widetilde{N} [\sigma_{x} (\cdot) \sigma_{x}]$ where $\alpha=\sum_{n=m}  p_{nm}$, meaning there is an allowed operation converting $\mathcal{N}$ into $\mathcal{M}$, completing the proof. 
\end{proof}

\textit{\textbf{Appendix on Results 3}} 
Here, we provide a sketch of the proof of Result~\ref{result3}. A complete derivation can be found in Supplementary Material D \cite{suppMat}.
\begin{proof}
The Choi-states of Weyl-covariant channels are diagonal in the Bell basis, and Choi-states that are diagonal in the Bell basis correspond to Weyl-covariant channels. The physical meaning of all the Choi-states of qubit unital channels being diagonalisable under local unitaries (see Lemma \ref{lemma1} in the above Appendix) can now be seen - all qubit unital channels are equal to a Weyl-covariant channel up to a pre-processing and post-processing unitary (This result was also recently found in \cite{li2023unital}, see Supplementary Material A \cite{suppMat} for our proof). From the diagonalisation of the Choi-state, the proof of Result~\ref{result3} is similar to that of Result~\ref{Result1} seen in the above Appendix. The difference arises from $\mathbb{B}$ now being in the convex hull of matrices with elements $B_{nm, kl} = \abs{ \bra{\Phi_{kl}} W_{ab} \otimes W_{cd} \ket{\Phi_{nm}} }^{2}$, where $W_{ab}$ and $W_{cd}$ are general discrete Weyl operators. This restricts the set of permutations to be only the set of local \textit{cyclic} permutations. A complete derivation can be found in Supplementary Material D \cite{suppMat}.
\end{proof}

{\color{black} \textit{\textbf{Appendix on Classical Capacity Quantifying Informational Non-equilibrium Preservability}}
The dynamical resource theory of classical communication \cite{Takagi_2020} has all state-preparation channels as the set of channels with no resource (the free set). This is due to these channels not being able to communicate any information. In the resource theory presented here, we have an additional thermodynamic constraint that the channels cannot drive the system out of equilibrium. This suggests that, conceptually, informational non-equilibrium preservability can be viewed as the ability to transmit classical information \cite{Takagi_2020, Hsieh_2021_Pres} subject to additional thermodynamic constraints. }



\bibliographystyle{apsrev4-1}
\bibliography{mainTextBib}

\appendix

\onecolumngrid

\section{Supplementary Material}

\section{\label{Appendix One}Supplementary Material A: The allowed operations of $\dynamicRT$}
The allowed operations of resource preservability theories, $\freeD$, presented in \cite{Hsieh_2020}, consist of a super-channel with a free channel in the memory system, see Fig.~\ref{fig1.5:freeOperations}.~a. When considering the underlying static resource theory to be the resource theory of informational non-equilibrium, the free super channels can be reformulated so that the memory system can be ignored. 

$\dynamicRT$ has only one free channel - the state preparation channel of the maximally mixed state, $\Lambda(\rho) = 1/d_{S} \textrm{Tr}(\rho) \mathbb{I}$. Therefore, an allowed operation $\Pi \in \freeD$, such that $\mathcal{M}=\Pi(\mathcal{N})$, where  $\mathcal{M} : \mathcal{K}_{in} \rightarrow \mathcal{K}_{out} \in \freeS, \mathcal{N}: \mathcal{H}_{in} \rightarrow \mathcal{H}_{out} \in \freeS$ can be written as 
\begin{equation}
    \mathcal{M} = \mathcal{E} \circ (\mathcal{N} \otimes \Lambda) \circ \mathcal{P}, \label{completeFreeOperation}
\end{equation}
where $\mathcal{P}: \mathcal{K}_{in} \rightarrow \mathcal{H}_{in} \otimes \mathcal{H}_{a}, ~ \mathcal{N} \otimes \Lambda: \mathcal{H}_{in} \otimes \mathcal{H}_{a} \rightarrow \mathcal{H}_{out} \otimes \mathcal{H}_{a}$ and $\mathcal{E}: \mathcal{H}_{out} \otimes \mathcal{H}_{a} \rightarrow \mathcal{K}_{out}$.

Consider a state $\rho \in \mathcal{K}_{in}$ input into the channel $\mathcal{M}$. The state after the application of $\mathcal{P}$ is $\sigma = \mathcal{P}(\rho) \in \mathcal{H}_{in} \otimes \mathcal{H}_{a}$. The memory system is then traced out by $\Lambda$, leaving $\sigma' = \textrm{Tr}_{a}( \sigma ) \in \mathcal{H}_{in}$. Anything stored in the memory space is then lost. A new channel can be defined as $\mathcal{P}' = \textrm{Tr}_{a}( \mathcal{P} ) : \mathcal{K}_{in} \rightarrow \mathcal{H}_{in}$, such that $\sigma' = \mathcal{P}'(\rho)$. This channel can be seen to be a unital channel still and hence is still a noisy operation.   

The state $\sigma'$ is then input into the channel $\mathcal{N}$ giving $\omega = \mathcal{N}(\sigma')$. A maximally mixed state is then appended to $\omega$ and it is input into a channel $\mathcal{E}$, see Fig.~\ref{fig1.5:freeOperations}.b. This appended maximally mixed state can be absorbed into the description of the channel $\mathcal{E}$ by considering it in the noisy operation formulation, as seen in Eq.~\eqref{noisyOperations}. We can then define a new unital channel $\mathcal{E}' : \mathcal{H}_{out} \rightarrow \mathcal{K}_{out}$, see Fig.~\ref{fig1.5:freeOperations}.c. 

Given we are considering channels with equal input and output dimension, such that $\mathcal{H}_{in}=\mathcal{H}_{out}$, all spaces become equal, 
\begin{equation}
    \mathcal{K}_{in} = \mathcal{K}_{out} = \mathcal{H}_{in}.
\end{equation}

The allowed operation given in Eq.~\eqref{completeFreeOperation} can now be written as 
\begin{equation}
    \mathcal{M} = \mathcal{E}' \circ \mathcal{N} \circ \mathcal{P}', 
\end{equation}
where $\mathcal{P',N,E'}: \mathcal{H}_{in} \rightarrow \mathcal{H}_{in} \in \freeS$. The channels can then be relabelled to $\mathcal{E,P}$. 

When a random variable, $\kappa$, is shared between the pre-processing and post-processing channels with some probability $p_{\kappa}$, the channels applied depend on $\kappa$. The allowed operations become 
\begin{equation}
    \mathcal{M} = \sum_{\kappa} p_{\kappa} \mathcal{E}_{\kappa} \circ \mathcal{N} \circ \mathcal{P}_{\kappa},
\end{equation}
where $\mathcal{E}_{\kappa}, \mathcal{P}_{\kappa} \in \freeS$. $\freeD$ now forms a convex set. 

The Choi-state of a channel $\mathcal{M}$ is then, 
\begin{equation}
    \begin{split}
        \choi{M} &= \big( \mathcal{I} \otimes \sum_{\kappa} p_{\kappa} \mathcal{E}_{\kappa} \circ \mathcal{N} \circ \mathcal{P}_{\kappa} \big) \ketbra{\Phi_{00}} \\
        &= \sum_{\kappa} p_{\kappa} (\mathcal{I} \otimes \mathcal{E}_{\kappa})(\mathcal{I} \otimes \mathcal{N})(\mathcal{I} \otimes \mathcal{P}_{\kappa})\ketbra{\Phi_{00}} \\
        &= \sum_{\kappa} p_{\kappa} (\mathcal{P}^{t}_{\kappa} \otimes \mathcal{I})(\mathcal{I} \otimes \mathcal{E}_{\kappa})(\mathcal{I} \otimes \mathcal{N})\ketbra{\Phi_{00}} \\
        &= \sum_{\kappa} p_{\kappa} (\Tilde{\mathcal{P}}_{\kappa} \otimes \mathcal{E}_{\kappa} )(\choi{N}), 
    \end{split}
\end{equation}
where we define $\mathcal{P}^{t} := \sum_{i} K_{i}^{t}(\cdot) K_{i}^{\dagger, t}$, which is the transpose of a channel $\mathcal{P}$ with Kraus representation $\mathcal{P} = \sum_{i} K_{i}(\cdot)K_{i}^{\dagger}$, where $(\cdot)^{t}$ is the transpose operation. In the penultimate line we have relabelled $\mathcal{P}^{t}_{\kappa} \rightarrow \Tilde{\mathcal{P}}_{\kappa}$ given the transpose of a unital channel is still a unital channel. This gives the Choi-state formulation of the allowed operations as defined in Eq.~\eqref{Choi State allowed operations} in the main text. 
\begin{figure}
    \centering
    \includegraphics[scale=0.5]{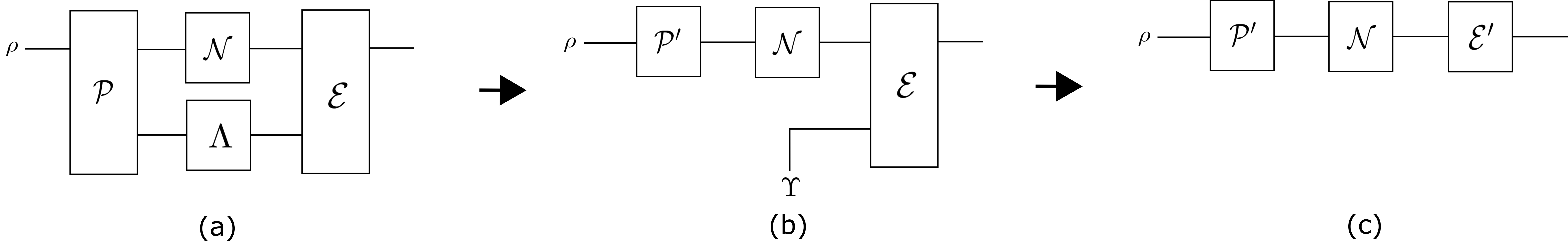}
    \caption{A quantum circuit diagram to depict how the allowed operations of resource preservability theories, defined in \cite{Hsieh_2020}, become the allowed operations of the dynamical resource theory of informational non-equilibrium defined in the main text.}
    \label{fig1.5:freeOperations}
\end{figure}

\section{\label{Appendix Two}Supplementary Material B: Qubit Unital Channels}
Result~\ref{Result1} states that an allowed operation of $\dynamicRT$ exists between a qubit channel $\mathcal{N}$ and qubit channel $\mathcal{M}$ if and only if a vector of eigenvalues of the Choi-state of $\mathcal{M}$ is in the convex hull of local permutations of a vector of eigenvalues of the Choi-state of $\mathcal{N}$. A full derivation is provided here. 

\subsection{Proof of Statement 1 implies statement 2 (Result~\ref{Result1})}
We begin by showing that statement 1 of Result~\ref{Result1} implies statement 2. We then show that statement 2 implies statement 1. 
\subsection{Proof of Lemma \ref{lemma1}}

\begin{proof}
Given, $\choi{N}$ is a valid bipartite density operator it can always be written in the block form as 
\begin{equation}
    \choi{N} = \frac{1}{4} \big( \mathbb{I} \otimes \mathbb{I} + (\bm{\beta} \cdot \bm{\sigma}) \otimes \mathbb{I} + \mathbb{I} \otimes (\bm{\gamma} \cdot \bm{\sigma}) + \sum_{ij} T_{ij} \sigma_{i} \otimes \sigma_{j} \big),
\end{equation}
where $\bm{\beta}, \bm{\gamma} \in \mathbb{R}^{3}$, $\bm{\sigma}$ is a vector of Pauli operators $\{ \sigma_{i} \}_{i=1}^{3}$ defined by 
\begin{equation}
    \sigma_{1} := \ket{0}\bra{1} + \ket{1}\bra{0},~ \sigma_{2} := i\ket{0}\bra{1} - i\ket{1}\bra{0},~ \sigma_{3} := \ket{0}\bra{0} - \ket{1}\bra{1}, \label{PauliOperators}
\end{equation}
in a given computation basis $\{ \ket{0}, \ket{1} \}$ and $T_{ij}$ are the matrix elements of a real matrix $T \in \mathbb{R}^{4} \otimes \mathbb{R}^{4}$. 

From the definition of $\choi{N}$ it is clear that the first marginal must be the maximally mixed state. Given $\mathcal{N}$ is a unital channel, the second marginal must also be the maximally mixed state, $\textrm{Tr}_{in}(\choi{N}) = \textrm{Tr}_{out}(\choi{N}) = \frac{1}{2}\mathbb{I}$. Using these conditions it can be seen that $\bm{\beta} \cdot \bm{\sigma} = \bm{\gamma} \cdot \bm{\sigma} = 0$.

It was shown in \cite{Horodecki_1996} that the matrix $T$, with elements $T_{ij}$, can always be digonalised under local unitary operations. This is due to every unitary transformation, $U$, having a uniquely associated rotation $O$ such that 
\begin{equation}
    U \bm{\beta} \cdot \bm{\sigma} U^{\dagger} = (O\bm{\beta}) \cdot \bm{\sigma}.
\end{equation}
Hence, local unitaries on $\choi{N}$ such that,   
\begin{equation}
    \begin{split}
        \choi{\Tilde{N}} &=  (U_{1} \otimes U_{2}) \choi{N}  (U_{1} \otimes U_{2})^{\dagger},
    \end{split}
\end{equation}
act to change the matrix $T$ to a matrix $T'$, where 
\begin{equation}
    T' = O_{1}TO^{t}_{2},
\end{equation}
where $O_{1}, O_{2}$ are rotations associated to the unitaries $U_{1}, U_{2}$ respectively, and $(\cdot)^{t}$ is the transpose operator. Due to the singular value decomposition \cite{banerjee2014linear}, two unitaries can therefore always be found that cause two rotations on $T$ that lead to $T'$ being diagonal. The Choi-state of qubit unital channels under local unitaries therefore become
\begin{equation}
    \choi{\Tilde{N}} = \frac{1}{4} \big( \mathbb{I} \otimes \mathbb{I} + \sum_{i} T_{ii} ~ \sigma_{i} \otimes \sigma_{i} \big). 
\end{equation}
In the standard basis $\choi{\Tilde{N}}$ is then
\begin{equation}
    \choi{\Tilde{N}} = \frac{1}{4} \begin{pmatrix}
        1+T_{33} & 0 & 0 & T_{11} - T_{22} \\
        0 & 1-T_{33} & T_{11} + T_{22} & 0 \\
        0 & T_{11} + T_{22} & 1 - T_{33} & 0 \\
        T_{11} - T_{22} & 0 & 0 & 1 + T_{33} \\
    \end{pmatrix}.
\end{equation}
This has the same form as a matrix that is diagonal in the Bell basis but written in the standard basis,
\begin{equation}
    \rho = \begin{pmatrix}
        p_{00} + p_{01} & 0 & 0 & p_{00} - p_{01} \\
        0 & p_{10} + p_{11} & p_{10} - p_{11} & 0 \\
        0 & p_{10} - p_{11} & p_{10} + p_{11} & 0 \\
        p_{00} - p_{01} & 0 & 0 & p_{00} + p_{01} \\
    \end{pmatrix}
\end{equation} 
where $\rho$ can equivalently be written as  
\begin{equation}
    \rho = \sum_{ij} p_{ij} \ketbra{\Phi_{ij}},
\end{equation}
where the states $\ket{\Phi_{nm}}$ are the basis states of the Bell basis generated from $\ket{\Phi_{00}}=1 / \sqrt{d_{S}} \sum_{i} \ket{ii}$ and the discrete Weyl operators \cite{watrous_2018} such that
\begin{equation}
    \ket{\Phi_{nm}} = \mathbb{I} \otimes W_{nm} \ket{\Phi_{00}},
\end{equation}
with 
\begin{equation}
    W_{nm} = \sum_{c=0}^{1} \Omega^{m c} \ket{n + c~\textrm{mod} ~2}\bra{c}, \hspace{2mm} \Omega = e^{\pi i}, \hspace{2mm} n,m \in \{0,1\}. \label{discreteWEylQubit}
\end{equation}
Hence, $\choi{\Tilde{N}}$ can alternatively be written as
\begin{equation}
    \choi{\Tilde{N}} = \sum_{nm} \lambda_{nm} \ketbra{\Phi_{nm}},
\end{equation}
where $\lambda_{ij}$ are the positive, real eigenvalues of $\choi{N}$ with $\sum_{nm} \lambda_{nm} = 1$, for which there exists a constructive form for their relationship with the matrix elements $\{T_{ii}\}_{1}^{3}$. 

\end{proof}

As mentioned in the main text, all unitary channels are in $\freeD$ and hence an allowed operation always exists that can diagonalise the Choi-states. The action of a channel within $\staticRT$ can therefore be fully captured through only the eigenvalues of its Choi-state. All of the Choi-states of channels in $\freeS$ are henceforth considered to be diagonal in the Bell basis. The tidle is now dropped when referring to channels in the Bell basis.  

By returning to the form of operations given in Eq.~\eqref{Choi State allowed operations} and using the above lemma, the allowed operations can be reformulated to be only in terms of the eigenvalues of the Choi-states. This greatly reduces the number of parameters in consideration. (We ignore the shared randomness in the follow proof for clarity, but the same result is reached if it is included). 

We remark that given we are considering all channels to have equal input and output dimension, we only need to consider mixed unitary channels. Consider the general pre-processing and post-processing mixed unitary channels,
\begin{equation}
    \mathcal{P}(\rho) = \sum_{i} p_{i} U_{i}\rho U_{i}^{\dagger}, \hspace{5mm}  \mathcal{E}(\rho) = \sum_{j} q_{j} V_{j}\rho V_{j}^{\dagger}.
\end{equation}
The Choi-state of a channel $\mathcal{M}$ under an allowed operation is given by
\begin{equation}
    \begin{split}
        \choi{M} &= \sum_{i,j} p_{i} q_{j} (U_{i}^{t} \otimes V_{j}) \choi{N} (U_{i}^{t} \otimes V_{j})^{\dagger} \\
        &= \sum_{i,j} p_{i} q_{j} (U_{i}^{t} \otimes V_{j}) \Big[ \sum_{nm} \lambda_{nm} \ketbra{\Phi_{nm}} \Big] (U_{i}^{t} \otimes V_{j})^{\dagger} \\
        &= \sum_{i,j} \sum_{nm} p_{i} q_{j} \lambda_{nm} (U_{i}^{t} \otimes V_{j}) (\ketbra{\Phi_{nm}}) (U_{i}^{t} \otimes V_{j})^{\dagger},
    \end{split}
\end{equation}
where the diagonal form of $\choi{N}$ has been used. 

Given operations in $\freeD$ map channels in $\freeS$ to channels in $\freeS$, $\choi{M}$ is the Choi-state of a qubit unital channel and can therefore also be written diagonally in the Bell basis,
\begin{equation}
    \choi{M} = \sum_{kl} \mu_{kl} \ketbra{\Phi_{kl}},
\end{equation}
with
\begin{equation}
    \mu_{kl} = \sum_{i,j} \sum_{nm} p_{i} q_{j} \lambda_{nm} 
    \abs{\bra{\Phi_{kl}} U_{i} \otimes V_{j} \ket{\Phi_{nm}} }^{2}.
\end{equation}
This can be interpreted as the product of a matrix and a vector, outputting another vector with components $\mu_{kl}$. 

The existence of an allowed operation in $\freeD$ can therefore be captured by the existence of a matrix $\mathbb{B}^{ij}$ such that 
\begin{equation}
    \bm{\mu} = \sum_{i,j} p_{i} q_{j} \mathbb{B}^{ij} \bm{\lambda}, \label{MidPointFormFreeOperationsAppendix}
\end{equation}
where $\bm{\lambda} \in \mathbb{R}^{4}$ is a vector of eigenvalues of $\choi{N}$, $\bm{\lambda}^{t} = [\lambda_{00}, \lambda_{01}, \lambda_{10}, \lambda_{11}]$, $\bm{\mu} \in \mathbb{R}^{4}$ is a vector of eigenvalues of the Choi-state $\choi{M}$, $ \bm{\mu}^{t} = [\mu_{00}, \mu_{01}, \mu_{10}, \mu_{11}]$, and $\mathbb{B}^{ij}$ is the matrix with elements $B^{ij}_{nm, kl} = \abs{ \bra{\Phi_{kl}} U_{i} \otimes V_{j} \ket{\Phi_{nm}} }^{2}$.

Allowed operations in $\freeD$ can now be succinctly written in terms of the eigenvalues of the Choi-states of channels in $\freeS$. To complete the proof of this direction, Lemma.~\ref{lemma2} must be proved.

The matrix $\mathbb{B}^{ij}$ is a product-Bell unistochastic matrix, meaning it is a doubly stochastic matrix that arises from squaring the absolute value of the matrix elements of a product unitary matrix written in the Bell basis. Given $\mathbb{B}^{ij}$ is doubly stochastic it can be written as a convex combination of permutation matrices \cite{bapat_raghavan_1997}. Here, we have the additional structure of the unitary matrix being a tensor product of two unitaries and a fixed basis. Lemma ~\ref{lemma2} shows that when considering Bell-product unistochastic matrices only the set of local permutations is needed in the convex combination. 

\subsection{Proof of Lemma \ref{lemma2}}
\begin{proof}
The unitaries $U_{i}$ and $V_{j}$ can be written in the Pauli operator basis, 
\begin{equation}
    U_{i} = \sum_{0}^{3} a^{i}_{\alpha} \sigma_{\alpha}, \hspace{5mm} V_{j} = \sum_{0}^{3} b^{j}_{\beta} \sigma_{\beta},
\end{equation}
with $\sigma_{0} := \mathbb{I}$, $\{ \sigma_{i} \}_{i=1}^{3}$ being the Pauli operators defined in Eq.~\eqref{PauliOperators} and $a_{\alpha}^{i}, b_{\beta}^{j} \in \mathbb{C} ~\forall ~(i,\alpha),(j,\beta)$. The matrix elements of $\mathbb{B}^{ij}$ in the Bell basis are then, 
\begin{equation}
    \begin{split}
        B^{ij}_{nm,kl} &= \abs{\bra{\Phi_{kl}} U_{i} \otimes V_{j} \ket{\Phi_{nm}} }^{2} \\
        &= \abs{\sum_{\alpha,\beta = 0}^{3} a^{i}_{\alpha}b^{j}_{\beta} \bra{\Phi_{kl}} (\sigma_{\alpha} \otimes \sigma_{\beta}) \ket{\Phi_{nm}}}^{2}.
    \end{split}
\end{equation}
Operators of the form $\mathbb{I} \otimes \sigma_{i} ~ \forall ~ i$ acting on Bell basis states permute them to other Bell basis state multiplied by a phase term. Ignoring the phase term for now, the action of $\mathbb{I} \otimes \sigma_{1}$ and $\sigma_{1} \otimes \mathbb{I}$ on the basis states is as follows, 
\begin{equation}
    \mathbb{I} \otimes \sigma_{1} \hspace{2mm} \& \hspace{2mm} \sigma_{1} \otimes \mathbb{I} : \hspace{3mm} \ket{\Phi_{00}} \rightarrow \ket{\Phi_{10}},~  \ket{\Phi_{01}} \rightarrow  \ket{\Phi_{11}},  ~  \ket{\Phi_{10}} \rightarrow  \ket{\Phi_{01}}, ~  \ket{\Phi_{11}} \rightarrow  \ket{\Phi_{00}}. 
\end{equation}
This permutation of the Bell basis states is referred to as $\perm{1}$, 
\begin{equation}
\renewcommand*{\arraystretch}{0.6}
    \perm{1} = \begin{pmatrix}
        0 & 0 & 1 & 0 \\
        0 & 0 & 0 & 1 \\
        1 & 0 & 0 & 0 \\
        0 & 1 & 0 & 0 
    \end{pmatrix} = \omega_{1} \otimes \omega_{0},
\end{equation}
where $\omega_{0}$ and $\omega_{1}$ are the qubit permutations, 
\begin{equation}
\renewcommand*{\arraystretch}{0.6}
    \omega_{0} = \begin{pmatrix}
1 & 0 \\
0 & 1 
\end{pmatrix}, \hspace{5mm} \omega_{1} = \begin{pmatrix}
0 & 1 \\
1 & 0 
\end{pmatrix}.
\end{equation}
Similarly we have 
\begin{equation}
    \begin{split}
        & \mathbb{I} \otimes \sigma_{2} \hspace{2mm} \& \hspace{2mm} \sigma_{2} \otimes \mathbb{I} : \hspace{3mm} \ket{\Phi_{00}} \rightarrow \ket{\Phi_{11}},~  \ket{\Phi_{01}} \rightarrow  \ket{\Phi_{10}},  ~  \ket{\Phi_{10}} \rightarrow  \ket{\Phi_{01}}, ~  \ket{\Phi_{11}} \rightarrow  \ket{\Phi_{00}}. \\
        & \mathbb{I} \otimes \sigma_{3} \hspace{2mm} \& \hspace{2mm} \sigma_{3} \otimes \mathbb{I} : \hspace{3mm} \ket{\Phi_{00}} \rightarrow \ket{\Phi_{01}},~  \ket{\Phi_{01}} \rightarrow  \ket{\Phi_{00}},  ~  \ket{\Phi_{10}} \rightarrow  \ket{\Phi_{11}}, ~  \ket{\Phi_{11}} \rightarrow  \ket{\Phi_{10}}. 
    \end{split}
\end{equation}
which we call $\perm{2}$ and $\perm{3}$ respectively where
\begin{equation}
\renewcommand*{\arraystretch}{0.6}
    \perm{2} = \begin{pmatrix}
        0 & 0 & 0 & 1 \\
        0 & 0 & 1 & 0 \\
        0 & 1 & 0 & 0 \\
        1 & 0 & 0 & 0
    \end{pmatrix} = \omega_{1} \otimes \omega_{1},
\end{equation}
and 
\begin{equation}
\renewcommand*{\arraystretch}{0.6}
    \perm{3} = \begin{pmatrix}
        0 & 1 & 0 & 0 \\
        1 & 0 & 0 & 0 \\
        0 & 0 & 0 & 1 \\
        0 & 0 & 1 & 0 \\
    \end{pmatrix} = \omega_{0} \otimes \omega_{1}.
\end{equation}
Additionally, there is $\mathbb{I} \otimes \mathbb{I}$ that maps all basis states to themselves giving $\perm{0} = \omega_{0} \otimes \omega_{0}$. This set of permutations, $\mathfrak{P} = \{ \perm{0}, \perm{1}, \perm{2}, \perm{3} \}$, are the local permutations given they can be broken down into a tensor product of qubit permutations. 

Combinations of permutations in $\mathfrak{P}$ are still in the set $\mathfrak{P}$. By splitting up the Pauli operators in $\mathbb{B}^{ij}$ they can be written as combinations of permutation in $\mathfrak{P}$,
\begin{equation}
    \begin{split}
        B^{ij}_{nm,kl} &= \abs{\sum_{\alpha,\beta = 0}^{3} a^{i}_{\alpha}b^{j}_{\beta} \bra{\Phi_{kl}} (\sigma_{\alpha} \otimes \mathbb{I})(\mathbb{I} \otimes \sigma_{\beta}) \ket{\Phi_{nm}}}^{2}, \\
        &= \abs{\sum_{\alpha,\beta = 0}^{3} a^{i}_{\alpha}b^{j}_{\beta} \theta_{\alpha, \beta} \bra{\Phi_{kl}} \perm{\alpha} \circ \perm{\beta} \ket{\Phi_{nm}}}^{2} \\
        &= \abs{\sum_{\gamma=0}^{3} Q^{ij}_{\gamma} \bra{\Phi_{kl}} \perm{\gamma} \ket{\Phi_{nm}}}^{2}, \\
        &= \abs{\bra{\Phi_{kl}} Q^{ij}_{0}\perm{0} + Q^{ij}_{1}\perm{1} + Q^{ij}_{2}\perm{2} + Q^{ij}_{3}\perm{3} \ket{\Phi_{nm}}}^{2}
    \end{split}
\end{equation}
where $\theta_{\alpha, \beta}$ is a phase factor arising from the permutations and in the third line a relabelling has taken place for clarity. The matrix $\mathbb{B}^{ij}$ can now be seen to be a combination of the local permutation matrices,
\begin{equation}
    \begin{split}
        \renewcommand*{\arraystretch}{1}
    \mathbb{B}^{ij} &= \begin{pmatrix}
        \abs{Q_{0}}^{2} & \abs{Q_{3}}^{2} & \abs{Q_{1}}^{2} & \abs{Q_{2}}^{2} \\
        \abs{Q_{3}}^{2} & \abs{Q_{0}}^{2} & \abs{Q_{2}}^{2} & \abs{Q_{1}}^{2} \\
        \abs{Q_{1}}^{2} & \abs{Q_{2}}^{2} & \abs{Q_{0}}^{2} & \abs{Q_{3}}^{2} \\
        \abs{Q_{2}}^{2} & \abs{Q_{1}}^{2} &\abs{Q_{3}}^{2} & \abs{Q_{0}}^{2}\\
    \end{pmatrix}
    \end{split}
\end{equation}
which can be written as
\begin{equation}
    \begin{split}
    \mathbb{B}^{ij} &= \abs{Q_{0}}^{2} \perm{0} + \abs{Q_{1}}^{2} \perm{1} + \abs{Q_{2}}^{2} \perm{2} + \abs{Q_{3}}^{2} \perm{3}, \\
    &=  \abs{Q_{0}}^{2} (\omega_{0} \otimes \omega_{0}) + \abs{Q_{1}}^{2} (\omega_{1} \otimes \omega_{0}) + \abs{Q_{2}}^{2} (\omega_{1} \otimes \omega_{1}) + \abs{Q_{3}}^{2} (\omega_{0} \otimes \omega_{1}).
    \end{split}
\end{equation}
where the $ij$ superscript on $Q_{\gamma}$ has been dropped to improve readability.

To be a convex combination it must be the case that 
\begin{equation}
    \sum_{\gamma=0}^{3} \abs{Q_{\gamma}}^{2} = 1.
\end{equation}
It can be seen from the matrix form of $\mathbb{B}^{ij}$ that a sum over all $Q_{\gamma}$ is a sum over the first column of the matrix,
\begin{equation}
    \begin{split}
        & \sum_{\gamma} \abs{Q_{\gamma}}^{2} = \sum_{n=0,m=0}^{1} B_{nm,00} \\
        &= \sum_{n=0,m=0}^{1} \abs{\bra{\Phi_{00}} U_{i} \otimes V_{j} \ket{\Phi_{nm}} }^{2} \\
        &= \sum_{n=0,m=0}^{1} \bra{\Phi_{00}} ~ (U_{i} \otimes V_{j}) (\ketbra{\Phi_{nm}}) (U_{i} \otimes V_{j})^{\dagger} ~ \ket{\Phi_{00}} \\
        &= \bra{\Phi_{00}} ~ (U_{i} \otimes V_{j}) \Bigg[ \sum_{n=0,m=0}^{1}\ketbra{\Phi_{nm}} \Bigg] (U_{i} \otimes V_{j})^{\dagger} ~ \ket{\Phi_{00}} \\
        &= 1.
    \end{split}
\end{equation}
as $\sum_{n=0,m=0}^{1}\ketbra{\Phi_{nm}} = \mathbb{I}$. Therefore, the matrix $\mathbb{B}^{ij}$ is a convex combination of local permutations. The introduction of shared randomness does not effect this result as a convex combination of product-Bell unistocastic matrices remains a convex combination of local permutation matrices. 
\end{proof}

Inputting all this into the form of allowed operations found above leads to the conclusion that an allowed operation existing between $\mathcal{N} \rightarrow \mathcal{M}$ (statement 1) implies that there exists a matrix between $\bm{\lambda}$ and $\bm{\mu}$ that is a convex combination of local permutations (statement 2).

To complete the proof of Result 1, the opposite direction must be shown. However, we will first provide the physical intuition behind Lemma \ref{lemma1} as it will be necessary to complete the proof of Result 1. 

\subsection{Lemma \ref{lemma1} Physical Intuition} 
A channel $\mathcal{M}$ is Weyl-covariant if $\mathcal{M}(W_{ij}\rho W^{\dagger}_{ij}) = ~ W_{ij}\mathcal{M}(\rho)W^{\dagger}_{ij}~\forall~ i,j, \rho$, where $W_{ij}$ are the discrete Weyl operators. Details on general Weyl-covariant channels can be found in Supplementary Material D and \cite{watrous_2018}, a brief overview of the necessary details will be given here.  

Any Weyl-covariant channel can be written in the following form: 
\begin{equation}
    \mathcal{M}(\rho) = \sum p_{ij} W_{ij} \rho W_{ij}^{\dagger},
\end{equation}
for some probability vector $p_{ij}$. The discrete Weyl operators in the qubit case, formally defined in Eq.~\eqref{discreteWEylQubit}, are $\{\mathbb{I}, \sigma_{x}, \sigma_{z}, -i\sigma_{y} \}$ where the Pauli operators are defined as in Eq.~\eqref{PauliOperators}. Therefore, a qubit Weyl-covariant channel can be always be written as 
\begin{equation}
    \mathcal{M}(\rho) = \sum_{k} p_{k} \sigma_{k} \rho \sigma_{k},
\end{equation}
for some probability vector $p_{k}$. 

\begin{lemma}
    Every qubit unital channel can be transformed into a Weyl covariant channel (a mixture of Pauli operators) via a pre-processing and post-processing unitary channel. 
\end{lemma}

\begin{proof}
Consider a qubit unital channel $\mathcal{N}$ with Choi-state 
\begin{equation}
    \choi{N} = (\mathcal{I} \otimes \mathcal{N})(\ketbra{\Phi_{00}}). 
\end{equation}
From Lemma ~\ref{lemma1}, there exists two unitaries $U_{1}$ and $U_{2}$ such that $\choi{N}$ can be diagonalised in the Bell basis,
\begin{equation}
    \choi{\Tilde{N}} = (U_{1} \otimes U_{2}) \choi{N} (U_{1} \otimes U_{2})^{\dagger}.\label{choistateunderlocalunitaries}
\end{equation}
Therefore, $\choi{\Tilde{N}}$ can equivalently be written as 
\begin{equation}
     \choi{\Tilde{N}} = \sum_{ij} \lambda_{ij} \ketbra{\Phi_{ij}}.  
\end{equation}
Choi-states that are diagonal in the Bell basis correspond to Weyl-covariant channels. This is a result of discrete Weyl operators being used to generate Bell basis states, 
\begin{equation}
    \begin{split}
        \choi{\Tilde{N}} &= \sum_{ij} \lambda_{ij} (\mathcal{I} \otimes W_{ij})(\ketbra{\Phi_{00}})(\mathcal{I} \otimes W_{ij}), \\
        &= (\mathcal{I} \otimes \sum_{ij} \lambda_{ij} W_{ij}(\cdot)W_{ij})(\ketbra{\Phi_{00}}) \label{Choistatethatdiggiveweylcovariant}
    \end{split}
\end{equation}
Hence, $\Tilde{\mathcal{N}}$ is a Weyl-covariant channel given $\sum_{ij} \lambda_{ij} = 1$ and $\lambda_{ij}$ is real and positive for all $i,j$. $\choi{\Tilde{N}}$ is then the Choi-state of a Weyl-covariant channel,
\begin{equation}
    \choi{\Tilde{N}} = (\mathcal{I} \otimes \Tilde{\mathcal{N}} ) (\ketbra{\Phi_{00}}).
\end{equation}
Returning to Eq.~\eqref{choistateunderlocalunitaries}, $\choi{\Tilde{N}}$ can be written as 
\begin{equation}
    \begin{split}
        \choi{\Tilde{N}} &= (U_{1} \otimes U_{2}) (\mathcal{I} \otimes \mathcal{N}) (\ketbra{\Phi_{00}}) (U_{1} \otimes U_{2})^{\dagger} \\    
         &= (\mathcal{I} \otimes U_{2}) (\mathcal{I} \otimes \mathcal{N}) (\mathcal{I} \otimes U_{1}^{t}) (\ketbra{\Phi_{00}})(\mathcal{I} \otimes U_{1}^{t})^{\dagger} (\mathcal{I} \otimes U_{2})^{\dagger} \\   
         &= (\mathcal{I} \otimes \mathcal{U}_{2} \circ \mathcal{N} \circ \mathcal{U}_{1})(\ketbra{\Phi_{00}}) 
    \end{split}
\end{equation}
where $\mathcal{U}_{1}$ and $\mathcal{U}_{2}$ are unitary channels that apply the unitaries $U_{1}^{t}$ and $U_{2}$ respectively. Due to the Choi-Jamiołkowski isomorphism, it must be the case that 
\begin{equation}
    \Tilde{\mathcal{N}}(\rho) = \mathcal{U}_{2} \circ \mathcal{N} \circ \mathcal{U}_{1}(\rho).
\end{equation}
Therefore, the diagonalisation of the Choi-states of qubit unital channels in the Bell basis is equivalent to qubit unital channels being equal to a Weyl-covariant channel up to a pre-processing and post-processing unitary. Formally, for all qubit unital channels $\mathcal{N}$ it is possibile to find two unitaries $U$ and $V$ such that 
\begin{equation}
    U\mathcal{N}(V \rho V^{\dagger})U^{\dagger} = p_{0} \mathbb{I}\rho \mathbb{I} + p_{1} \sigma_{x} \rho \sigma_{x} + p_{2} \sigma_{y} \rho \sigma_{y} + p_{3} \sigma_{z} \rho \sigma_{z},
\end{equation}
where $\sum_{k} p_{k} = 1$. This result was recently also found in \cite{li2023unital}. 
\end{proof}

\subsection{Proof of Statement 2 implies Statement 1 (Result~\ref{Result1})}

We now return to the proof of result 1 and show that statement 2 implies statement 1. 

\begin{proof}
Assuming statement 2 holds, we can write
\begin{equation}
    \bm{\mu} = \sum_{nm} p_{nm} \sigma^{n}_{x} \otimes \sigma_{x}^{m} \bm{\lambda}.
\end{equation}
Translating back to the Choi-states of the channels these vectors represent, we find 
\begin{equation}
    \choi{\Tilde{M}} = \sum_{nm} p_{nm} (\sigma^{n}_{x} \otimes \sigma_{x}^{m}) (\choi{\Tilde{N}}) (\sigma^{n}_{x} \otimes \sigma_{x}^{m}),
\end{equation}
where $\choi{\Tilde{M}}$ and $\choi{\Tilde{N}}$ are Choi-states diagonal in the Bell basis with eigenvalues given by the elements of $\bm{\mu}$ and $\bm{\lambda}$ respectively. This can be re-written as 
\begin{equation}
    \begin{split}
        \choi{\Tilde{M}} &= \sum_{nm} p_{nm} (\sigma^{n}_{x} \otimes \sigma_{x}^{m}) (\mathcal{I} \otimes \mathcal{\Tilde{\mathcal{N}}})(\ketbra{\Phi_{00}}) (\sigma^{n}_{x} \otimes \sigma_{x}^{m}), \\
        &= \sum_{nm} p_{nm} (\mathcal{I} \otimes \sigma_{x}^{m}) (\mathcal{I} \otimes \mathcal{\Tilde{\mathcal{N}}})(\mathcal{I} \otimes \sigma^{n}_{x})(\ketbra{\Phi_{00}})(\mathcal{I} \otimes \sigma^{n}_{x}) (\mathcal{I} \otimes \sigma_{x}^{m}), \\
        &= \sum_{nm} p_{nm} (\mathcal{I} \otimes \sigma_{x}^{m} [\cdot] \circ \mathcal{\Tilde{\mathcal{N}}} \circ \sigma^{n}_{x} [\cdot])(\ketbra{\Phi_{00}}),
    \end{split}
\end{equation}
where $\sigma_{x}^{m} [\cdot]$ and $\sigma^{n}_{x} [\cdot]$ are channels applying $\sigma_{x}^{m}$ and $\sigma_{x}^{n}$ respectively, and we have used the fact that $\sigma_{x}^{t} = \sigma_{x}$. Hence, from the Choi-Jamiołkowski isomorphism, $\Tilde{\mathcal{M}}$ is given by 
\begin{equation}
    \begin{split}
        \Tilde{\mathcal{M}}(\rho) &= \sum_{nm} p_{nm} \sigma_{x}^{m} \circ \mathcal{\Tilde{\mathcal{N}}} \circ \sigma^{n}_{x} (\rho) \\
        &= p_{00} \mathcal{\Tilde{\mathcal{N}}}(\rho) + p_{01}\sigma_{x}\mathcal{\Tilde{\mathcal{N}}}(\rho)\sigma_{x} + p_{10}\mathcal{\Tilde{\mathcal{N}}}(\sigma_{x}\rho\sigma_{x}) + p_{11}\sigma_{x}\mathcal{\Tilde{\mathcal{N}}}(\sigma_{x}\rho\sigma_{x})\sigma_{x}.
    \end{split}
\end{equation}
In the previous section it was shown that $\Tilde{\mathcal{N}}$ corresponds to a Weyl-covariant channel given it is associated to a Choi-state diagonal in the Bell basis. In addition, $\sigma_{x}$ is a discrete Weyl operator in the qubit case and hence,
\begin{equation}
   \begin{split}
        \Tilde{\mathcal{M}}(\rho) &= (p_{00} + p_{11}) \mathcal{\Tilde{\mathcal{N}}}(\rho) + (p_{01} + p_{10}) \mathcal{\Tilde{\mathcal{N}}}(\sigma_{x}\rho\sigma_{x}) \\
    &= \Tilde{\mathcal{N}} \circ \mathcal{P},
   \end{split}
\end{equation}
where $\mathcal{P}(\rho) = (p_{00} + p_{11})\mathcal{I}(\rho) +  (p_{01} + p_{10}) \sigma_{x} \rho \sigma_{x}$. Therefore, if statement 2 holds then there always exists a mixed unitary pre-processing channel achieving that mapping $\mathcal{N} \rightarrow \mathcal{M}$. Hence, statement 2 implies statement 1. This completes the proof of Result \ref{Result1}. 
\end{proof}

\subsection{Depolarising and Dephasing Example}
Within this section Result~\ref{Result1} is applied to the depolarising and dephasing channels to provide an example of the application of the dynamical resource theory of informational non-equilibrium. Firstly, the existence of an allowed operation between two depolarising channels and then two dephasing channels is assessed. Secondly, the existence of an allowed operations between a depolarising and a dephasing channel is assessed, and vice versa. 
\subsubsection{Depolarising to Depolarising}
The depolarising channel, $\depola{s}$, is given by 
\begin{equation}
    \depola{s}(\rho) = s\mathcal{I}(\rho) + (1-s) \textrm{Tr}(\rho)\frac{\mathbb{I}}{d_{S}},
\end{equation}
where $0 \leq s \leq 1$. This channel is a mixed unitary channel for all values of $s$. By applying Result~\ref{Result1} we aim to show what other qubit depolarising channels ($d_{s}=2$) can be simulated from a qubit depolarising channel with a given $s$.  
\begin{lemma}
    There exists an allowed operation $\Pi \in \freeD$ such that $\depola{s'} = \Pi(\depola{s})$ if and only if $s' \leq s$.  
\end{lemma}
\begin{proof}
Firstly, the Choi-state of $\depola{s}$ needs to be found and diagonlised in the Bell basis. This gives
\begin{equation}
    \choi{\mathcal{D}^{\rm{pol}}} = \Big(s + \frac{1-s}{4} \Big) \ketbra{\Phi_{00}} + \Big(\frac{1-s}{s} \Big) \ketbra{\Phi_{01}} + \Big(\frac{1-s}{s} \Big) \ketbra{\Phi_{10}} +  \Big(\frac{1-s}{s} \Big) \ketbra{\Phi_{11}}.
\end{equation}
The coefficients of this decomposition are the eigenvalues of the Choi-state. These eigenvalues are then collected into a vector
\begin{equation}
    \bm{\lambda}^{\rm{pol}} = \begin{bmatrix}
        s +  (1-s)/4 \\
        (1-s)/4 \\
    (1-s)/4 \\
     (1-s)/4 \\
    \end{bmatrix}.
\end{equation}
The target channels we are considering are depolarising channels with a different value of $s$, hence 
\begin{equation}
    \bm{\mu}^{\rm{pol}} =  \begin{bmatrix}
        s' +  (1-s')/4 \\
        (1-s')/4 \\
    (1-s')/4 \\
     (1-s')/4 \\
    \end{bmatrix}.
\end{equation}
Result~\ref{Result1} states that an allowed operation $\depola{s} \rightarrow \depola{s'}$ exists if and only if there exists a matrix that is a convex combination of local permutations mapping $\bm{\lambda}^{\rm{pol}}$ to $\bm{\mu}^{\rm{pol}}$. Therefore, an allowed operation exists if there exists a set of values $\{p_{00}, p_{01}, p_{10}, p_{11} \}$, where $\sum_{ij} p_{ij} = 1$ and $p_{ij}$ is real and positive for all $(i,j)$, such that   
\begin{equation}
    \begin{bmatrix}
        s' +  (1-s')/4 \\
        (1-s')/4 \\
    (1-s')/4 \\
     (1-s')/4 \\
    \end{bmatrix} = \begin{bmatrix}
        p_{00}s +  (1-s)/4 \\
        p_{01}s + (1-s)/4 \\
    p_{10}s + (1-s)/4 \\
     p_{11}s + (1-s)/4 \\ 
    \end{bmatrix}.  \label{depolafreeop}
\end{equation}
The left hand side of Eq.~\eqref{depolafreeop} has been generated by applying the set of local qubit permutations to $\bm{\lambda}^{pol}$ and then taking the convex combination of the resultant vectors. This forms produces a set of four equations that, along with the condition $\sum_{ij} p_{ij} = 1$, can be used to find the convex combination coefficients in terms of $s$ and $s'$, giving
\begin{equation}
    \begin{split}
        p_{00} &= \frac{1}{4}\big(3\frac{s'}{s} + 1 \big), \\
        p_{01} &= p_{10} = p_{11} =  \frac{1}{4}\big(1-\frac{s'}{s} \big).
    \end{split}
\end{equation}
By enforcing the conditions that $0 \leq p_{ij} \leq 1 ~ \forall~ i,j$, it can be seen that there exists an allowed operation of $\dynamicRT$ mapping $\depola{s}$ to $\depola{s}$ if and only if $s' \leq s$. We note that the $s=0$ case is straightforward, and hence we focus on the case that $s>0$. 

\end{proof}

\subsubsection{Dephasing to Dephasing}
The dephasing channel, $\dephase{q}$ is given by 
\begin{equation}
    \dephase{q}(\rho) = q \mathcal{I}(\rho) + (1-q) \widetilde{\mathcal{D}}(\rho), 
\end{equation}
where $0 \leq q \leq 1$ and $\widetilde{\mathcal{D}}$ is the completely dephasing channel, 
\begin{equation}
    \widetilde{\mathcal{D}}(\rho) = \sum_{n} \ketbra{n} \rho \ketbra{n}. 
\end{equation}
This channel is mixed unitary for all values of $q$. We now prove the following Lemma. 
\begin{lemma}
    There exists an allowed operation $\Pi \in \freeD$ such that $\dephase{q'} = \Pi(\dephase{q})$ if and only if $q' \leq q$.   
\end{lemma}

\begin{proof}
The Choi-state of $\dephase{q}$ diagonlised in the Bell basis is given by 
\begin{equation}
    \choi{\mathcal{D}\rm{ph}} = \frac{1}{2}(1+q) \ketbra{\Phi_{00}} + \frac{1}{2}(1-q)\ketbra{\Phi_{01}},
\end{equation}
giving the vector of eigenvalues, $\bm{\lambda}^{\rm{ph}}$, to be a
\begin{equation}
    \bm{\lambda}^{\rm{ph}} = \frac{1}{2}\begin{bmatrix}
        (1+q) \\
        (1-q) \\
        0 \\
        0 
    \end{bmatrix}.
\end{equation}
The target channel $\bm{\mu}^{\rm{ph}}$ is then 
\begin{equation}
    \hspace{0.5cm} \bm{\mu}^{\rm{ph}} = \frac{1}{2}\begin{bmatrix}
        (1+q') \\
        (1-q') \\
        0 \\
        0 
    \end{bmatrix}.
\end{equation}
An allowed operation $\dephase{q} \rightarrow \dephase{q'}$ exists if and only if there exists a set of values $\{p_{00}, p_{01}, p_{10}, p_{11} \}$, where $\sum_{ij} p_{ij} = 1$ and $p_{ij}$ is real and positive for all $(i,j)$, such that 
\begin{equation}
    \frac{1}{2}\begin{bmatrix}
        (1+q') \\
        (1-q') \\
        0 \\
        0 
    \end{bmatrix} = \frac{1}{2}\begin{bmatrix}
        p_{00}(1+q) + p_{01}(1-q) \\
        p_{00}(1-q) + p_{01}(1+q) \\
        p_{10}(1-q) + p_{11}(1+q) \\
        p_{10}(1+q) + p_{11}(1-q)
    \end{bmatrix}.
\end{equation}
As above, this forms produces a set of four equations that, along with the condition $\sum_{ij} p_{ij} = 1$, can be used to find the convex combination coefficients in terms of $q$ and $q'$. They are found to be 
\begin{equation}
    \begin{split}
        p_{00} &= \frac{1}{2}(1+\frac{q'}{q}) \\
        p_{01} &= \frac{1}{2}(1-\frac{q'}{q}) \\
        p_{10} &= p_{11} = 0. 
    \end{split}
\end{equation}
By again enforcing the conditions that $0 \leq p_{ij} \leq 1 ~ \forall~ i,j$, it can be seen that there exists an allowed operation of $\dynamicRT$ mapping $\dephase{q}$ to $\dephase{q'}$ if and only if $q' \leq q$.  
\end{proof}

\subsubsection{Dephasing to Depolarising}
We now consider when there exists an allowed operation mapping a dephasing channel, $\dephase{q}$, to a depolarising channel, $\depola{s}$.  
\begin{lemma}
    There exists an allowed operation $\Pi \in \freeD$ such that $\depola{s} = \Pi(\dephase{q})$ if and only if $s \leq \frac{q}{2-q}$.   
\end{lemma}
\begin{proof}
An allowed operation $\dephase{q} \rightarrow \depola{s}$ exists if and only if there exists a set of values $\{p_{00}, p_{01}, p_{10}, p_{11} \}$, where $\sum_{ij} p_{ij} = 1$ and $p_{ij}$ is real and positive for all $(i,j)$, such that 
\begin{equation}
    \begin{bmatrix}
        s +  (1-s)/4 \\
        (1-s)/4 \\
    (1-s)/4 \\
     (1-s)/4 \\
    \end{bmatrix} = \frac{1}{2}\begin{bmatrix}
        p_{00}(1+q) + p_{01}(1-q) \\
        p_{00}(1-q) + p_{01}(1+q) \\
        p_{10}(1-q) + p_{11}(1+q) \\
        p_{10}(1+q) + p_{11}(1-q)
    \end{bmatrix}.
\end{equation}
This produces a set of four equations that, along with the condition $\sum_{ij} p_{ij} = 1$, can be used to find the convex combination coefficients in terms of $q$ and $s$. The convex combination coefficients can be found to be 
\begin{equation}
    \begin{split}
        p_{00} &= \frac{1}{4} \big(s + 1 + \frac{2s}{q} \big) \\
        p_{01} &= \frac{1}{4} \big(s + 1 - \frac{2s}{q} \big) \\
        p_{10} &= p_{11} = \frac{1}{4}(1-p). 
    \end{split}
\end{equation}
By enforcing $0 \leq p_{01} \leq 1$ in particular, it can be seen that there exists an allowed operation of $\dynamicRT$ mapping $\dephase{q}$ to $\depola{s}$ if and only if $s \leq q/(2-q)$. It is $p_{01}$ that provides the strictest inequalities on the relationship between $s$ and $q$. 
\end{proof}

\subsubsection{Depolarising to Dephasing}
Finally, there exists an allowed operation mapping a depolarising channel, $\depola{s}$, to a dephasing channel, $\dephase{q}$, is considered.  

\begin{lemma}
    There exists an allowed operation $\Pi \in \freeD$ such that $\dephase{q} = \Pi(\depola{s})$ if and only if $s=1$.   
\end{lemma}
\begin{proof}
An allowed operation $\depola{s} \rightarrow \dephase{q}$ exists if and only if there exists a set of values $\{p_{00}, p_{01}, p_{10}, p_{11} \}$, where $\sum_{ij} p_{ij} = 1$ and $p_{ij}$ is real and positive for all $(i,j)$, such that 
\begin{equation}
    \frac{1}{2}\begin{bmatrix}
        (1+q) \\
        (1-q) \\
        0 \\
        0 
    \end{bmatrix} = \begin{bmatrix}
        p_{00}s +  (1-s)/4 \\
        p_{01}s + (1-s)/4 \\
    p_{10}s + (1-s)/4 \\
     p_{11}s + (1-s)/4 \\
    \end{bmatrix}.
\end{equation}
It can be seen immeditaly that the final two terms can only be zero if and only if $s=1$. Hence, there exists an allowed operation of $\dynamicRT$ mapping $\depola{s}$ to $\dephase{q}$ if and only if $s = 1$. This is the trivial case in which $\depola{s}(\rho) = \mathcal{I}(\rho)$ and is not a depolarising channel. We can therefore conclude that no allowed operation exists that can map any depolarising channel to any dephasing channel. 
\end{proof}

Physically, the above results say that if $s \leq q/(2-q)$ then a dephasing channel $\dephase{q}$ preserves purity better on \textit{all} input states than $\depola{s}$. However, for all values of $s$ and $q$, apart from the trivial case, there will be some input states for which $\dephase{q}$ will preserve purity better than $\depola{s}$.  

\section{\label{Appendix Three}Supplementary Material C: A complete Set of Monotones}
Result~\ref{result2} and Result~\ref{result4} states that if the vectors in $\Tilde{{\mathfrak{L}}}$ are linearly independent they form a simplex in $\mathbb{R}^{d_{S}^{2}-1}$ dimensional space and a set of $d_{S}^{2}$ inequalities can be found that are a complete set of monotones. If the vectors in $\Tilde{{\mathfrak{L}}}$ are not linearly independent then varying amounts of monotones are needed to make a complete set. 

As mentioned in the main text, in general these vector will be linearly independent. Any amount of noise included in the description of the channels will likely ensure that the subset of cases that are linearly dependent will be perturbed such that they are no longer linearly dependent. We give a description of both cases here for completeness.

Recall that $\mathfrak{L} = \{ \perm{i} \bm{\lambda} : \forall ~ \perm{i} \in \mathfrak{P} \}$ and that $\Tilde{\mathfrak{L}}$ are the vectors in $\mathfrak{L}$ with the final value in each vector discarded due to normalisation. In what follows, dimSpan($\cdot$) denotes the dimension of the space spanned by the set ($\cdot$). 
\begin{lemma} \label{lemma3}
\textit{If}
\begin{equation}
    \textup{dimSpan} (\mathfrak{K}) = d_{S}^{2} - 1, 
\end{equation}
\textit{where}
\begin{equation}
    \mathfrak{K} = \{ \Tilde{\bm{\lambda}}_{i} - \Tilde{\bm{\lambda}}_{0} : \Tilde{\bm{\lambda}}_{0} \in  \Tilde{\mathfrak{L}}, \forall ~ \Tilde{\bm{\lambda}}_{i} \in \Tilde{\mathfrak{L}} \}, 
\end{equation}
\textit{then $d_{S}^{2}$ inequalities are needed to create a complete set of monotones. Else if}
\begin{equation}
    \textup{dimSpan} (\mathfrak{K}) = k, k < d_{S}^{2} - 1,
\end{equation}
\textit{then $\eta$ inequalities and $\zeta$ equalities are needed to create a complete set of monotones, where $\eta$ and $\zeta$ are integers less than $d_{s}^{2}$.}
\end{lemma}

\begin{proof}
If $\textrm{dimSpan} (\mathfrak{K}) = d_{S}^{2} - 1$ the set of vectors $\Tilde{{\mathfrak{L}}}$ are linearly independent and form a simplex in $\mathbb{R}^{d_{S}^{2}-1}$. If $\Tilde{\bm{\mu}}$, the $\mathbb{R}^{d_{S}^{2}-1}$ representation of a channel $\bm{\mu} \in \mathbb{R}^{d_{S}^{2}}$, is within the simplex then it is in the convex hull of the vectors in $\Tilde{\mathfrak{L}}$. This is equivalent to $\bm{\mu}$ being in the convex hull of the vectors in $\mathfrak{L}$. If so, an allowed operation exists between the channels represented by $\bm{\lambda}$ and $\bm{\mu}$.

Each monotone is associated to one facet of the simplex and assess if $\Tilde{\bm{\mu}}$ is on the side of the facet that places it inside the simplex, or the side of the facet that places it outside. This can be achieved through the description of the hyper-planes that encompasses each facet. 

A hyper-plane is described by a vector $\bm{f} \in \mathbb{R}^{d_{S}^{2} - 1}$ that is perpendicular to all vectors in the hyper-plane. Any vector beginning at the origin and terminating on the hyper-plane, $\bm{a} \in \mathbb{R}^{d_{S}^{2} - 1}$, has an equal value dot product with $\bm{f}$, $\bm{f} \cdot \bm{a} = l$. Any vector that terminates before the hyper-plane, $\bm{b} \in \mathbb{R}^{d_{S}^{2} - 1}$, has a dot product with $\bm{f}$ less than $l$, $\bm{f} \cdot \bm{b} < l$. Similarly, any vector that terminates after the hyper-plane, $\bm{c} \in \mathbb{R}^{d_{S}^{2} - 1}$, has a dot product with $\bm{f}$ greater than $l$, $\bm{f} \cdot \bm{c} > l$. The value of the dot product of a vector with $\bm{f}$ therefore provides information about the geometry of the vector with respect to the plane $\bm{f}$ defines. See Fig.~~\ref{fig:2} for a graphical depiction in $\mathbb{R}^{3}$ of these vectors with respect to a plane given by $\bm{f}$.
\begin{figure}
    \centering
    \includegraphics[scale=0.4]{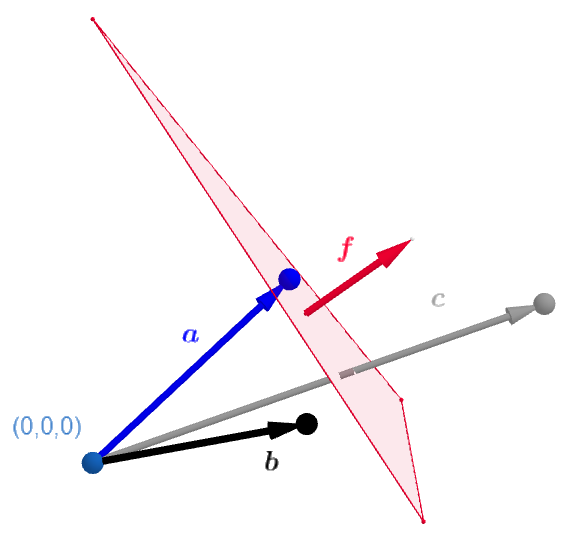}
    \caption{A section of a plane is shown in red. The plane is defined by the vector $\bm{f}$ that is perpendicular to all vector in the plane. The three vectors $\{ \bm{a}, \bm{b}, \bm{c} \}$ begin at the origin. The value of the dot product of these vectors with $\bm{f}$ is dependent on where these vectors determinate with respect to the plane described by $\bm{f}$. (Figure made with \cite{gg2}).}
    \label{fig:2}
\end{figure}

Each hyper-plane encompassing a facet of the simplex has all but one of the vectors in $\Tilde{{\mathfrak{L}}}$ terminate within it. A vector $\bm{f}$ can be found that is perpendicular to these $d_{S}^{2}-1$ vectors and $l$ can then be calculated by taking the dot product of $\bm{f}$ with one of these vectors. The value of $l$ can then be thought of as a distance like measure from the origin to one facet of the simplex. 

To decide what the `correct side' of the hyper-plane $\Tilde{\bm{\mu}}$ needs to be on is, such that it is inside the simplex, $\Tilde{\bm{\Lambda}} \cdot \bm{f}$ can be checked. Given $\Tilde{\bm{\Lambda}}$ is always in the centre of the simplex, the correct side of the plane to be on is the side $\Tilde{\bm{\Lambda}}$ is on. 

By calculating the value of $\bm{f} \cdot \Tilde{\bm{\mu}}$ and comparing it to $l$ and $\Tilde{\bm{\Lambda}} \cdot \bm{f}$, it can be decided if $\Tilde{\bm{\mu}}$ is on the same side of the plane, given by $\bm{f}$, as $\Tilde{\bm{\Lambda}}$.

A set of vectors $\mathfrak{f}$ can then be found where each element describes a hyper-plane encompassing one facet of the simplex. By repeating the above procedure for all elements in $\mathfrak{f}$, of which there are $d_{S}^{2}$, it can be decided if $\Tilde{\bm{\mu}}$ is in the convex hull of $\Tilde{\mathfrak{L}}$ using $d_{S}^{2}$ inequalities.

If $\textrm{dimSpan} (\mathfrak{K}) < d_{S}^{2} - 1$, the number of monotones needed will vary with the geometry of the problem, with both inequalities and equalities being needed. This is demonstrated through an example using qubit channels, $d_{S}=2$. 

If $\textrm{dimSpan} (\mathfrak{K}) = 2$, then the simplex formed by the vectors in $\Tilde{\mathfrak{L}}$ is a quadrilateral in $\mathbb{R}^{3}$. For an allowed operation to exists $\Tilde{\bm{\mu}}$ must sit in this quadrilateral. In this case one equality is needed to assess if $\Tilde{\bm{\mu}}$ is in the plane that contains this quadrilateral. Four inequalities are then needed to ensure $\Tilde{\bm{\mu}}$ sits on the correct side of the lines that make the boundaries of this quadrilateral. This leaves a total of four inequalities and one equality, $\eta=4$ and $\zeta=1$.

If the vectors in $\Tilde{\mathfrak{L}}$ form a line then it was found that two equalities and one inequality, $\eta=1$ and $\zeta=2$, were needed. Hence, when $\textrm{dimSpan} (\mathfrak{K}) < d_{S}^{2} - 1$ for general dimension channels the number of constraints needed will depend on the geometry created by $\Tilde{\mathfrak{L}}$ in a non-trivial way. 

\end{proof}

\subsection{Qubit Channels Example}

For clarity, an example of the formation of a complete set of monotones for $\dynamicRT$ when considering qubit channels, $d_{S}=  2$, is given. The complete set of monotones will determine if there exists an allowed operations $\Pi \in \freeD$, such that $\mathcal{M} = \Pi(\mathcal{N})$, where $\mathcal{M,N} \in \freeS$. We assume $\textrm{dimSpan} (\mathfrak{K}) = 3$ in the following.

The vector $\bm{\lambda} \in \mathbb{R}^{4}$ containts the eigenvalues of the Choi-state $\choi{N}$, $\bm{\lambda}^{t} = [\lambda_{00}, \lambda_{01}, \lambda_{10}, \lambda_{11}]$, and $\bm{\mu} \in \mathbb{R}^{4}$ is a vector of eigenvalues of the Choi-state $\choi{M}$, $ \bm{\mu}^{t} = [\mu_{00}, \mu_{01}, \mu_{10}, \mu_{11}]$. 

The set $\Tilde{\mathfrak{L}}$ is given by 
\begin{equation}
    \Tilde{\mathfrak{L}} = \Bigg{\{} \begin{bmatrix}
           \lambda_{00} \\
           \lambda_{01} \\
            \lambda_{10} 
         \end{bmatrix} , \begin{bmatrix}
           \lambda_{01} \\
           \lambda_{00} \\
            \lambda_{11} 
         \end{bmatrix} , \begin{bmatrix}
           \lambda_{11} \\
           \lambda_{10} \\
            \lambda_{01} 
         \end{bmatrix}, \begin{bmatrix}
           \lambda_{10} \\
           \lambda_{11} \\
            \lambda_{00} 
         \end{bmatrix} \Bigg{\}}, \label{setTildeL}
\end{equation}
with the simplex forming a tetrahedron when plotted in $\mathbb{R}^{3}$. A collection of planes where each one encompasses a face of the tetrahedron must be found. The first face considered is that which includes the first three vectors of $\Tilde{\mathfrak{L}}$. This is described by the following vector
\begin{equation}
    \bm{f}_{00} = \Bigg[  \begin{bmatrix}
           \lambda_{00} \\
           \lambda_{01} \\
            \lambda_{10} 
         \end{bmatrix} - \begin{bmatrix}
           \lambda_{01} \\
           \lambda_{00} \\
            \lambda_{11} 
         \end{bmatrix} \Bigg] \times \Bigg[  \begin{bmatrix}
           \lambda_{00} \\
           \lambda_{01} \\
            \lambda_{10} 
         \end{bmatrix} - \begin{bmatrix}
           \lambda_{11} \\
           \lambda_{10} \\
            \lambda_{01} 
         \end{bmatrix} \Bigg],
\end{equation}
and is shown in Fig.~\ref{fig3:operationalInterpretation}. This is the cross-product of two vectors within the plane, meaning $\bm{f}_{00}$ is a vector perpendicular to all vector in the plane. The value of $l$ can then be found by taking the dot product of $\bm{f}_{00}$ with any of the three vectors,
\begin{equation}
    \bm{f}_{00} \cdot  \begin{bmatrix}
           \lambda_{00} \\
           \lambda_{01} \\
            \lambda_{10} 
         \end{bmatrix} = \bm{f}_{00} \cdot  \begin{bmatrix}
           \lambda_{01} \\
           \lambda_{00} \\
            \lambda_{11} 
         \end{bmatrix} = \bm{f}_{00} \cdot   \begin{bmatrix}
           \lambda_{11} \\
           \lambda_{10} \\
            \lambda_{01} 
         \end{bmatrix} = l,
\end{equation}
as all vectors that terminate in the plane have the same dot product with the vector that defines it. It can be seen that there is some choice in how the value of $l$ is found, this gives some flexibility in the form of the monotone. 

\begin{figure}
    \centering
    \includegraphics[scale=0.3]{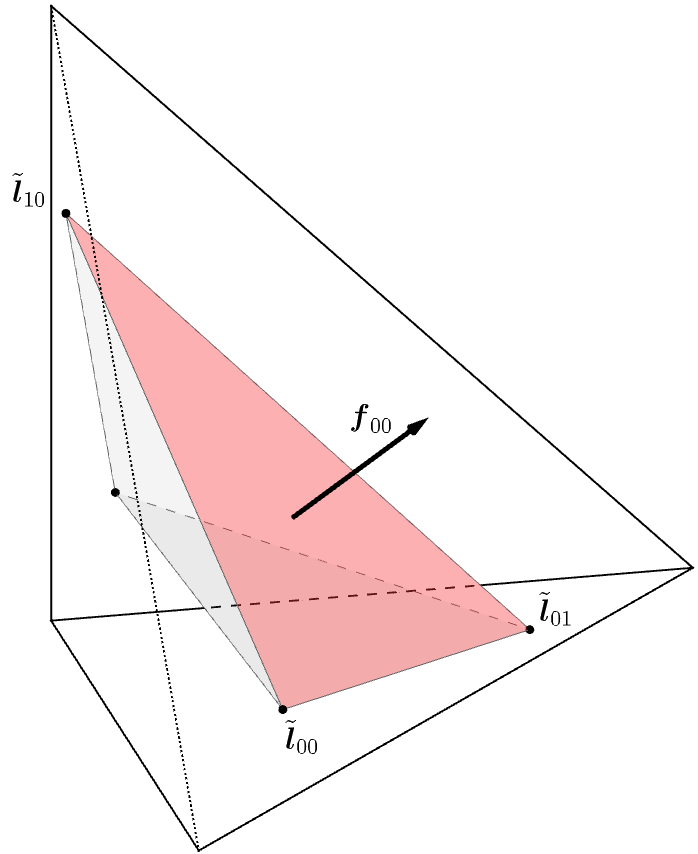}
    \caption{A single face of the simplex given by $\Tilde{\mathfrak{L}}$ is highlighted and the vector defining the plane encompassing that face, $\bm{f}_{00}$, is shown. The plane contains three of the four vectors in $\Tilde{\mathfrak{L}}$ (Figure made with \cite{gg2}).}
    \label{fig3:operationalInterpretation}
\end{figure}

By checking the dot product of $\bm{f}_{00}$ with $\Tilde{\bm{\Lambda}}$, the correct side of the face can be found. Here we assume 
\begin{equation}
    \bm{f}_{00} \cdot \Tilde{\bm{\Lambda}} = \bm{f}_{00} \cdot \begin{bmatrix}
           1/4 \\
           1/4 \\
        1/4 
         \end{bmatrix} < l.
\end{equation}
Hence, any vector $\bm{a} \in \mathbb{R}^{3}$ with $\bm{a} \cdot \bm{f}_{00} < l$ will be on the correct side of the face (The introduction of a minus sign where necessary can ensure this is always a desire order of inequaility $<$). 

The three dimensional representation of the channel $\bm{\mu}$ is given by $\Tilde{\bm{\mu}}^{t} = [\mu_{00}, \mu_{01}, \mu_{10}]$. One of the set of complete monotones is then given by 
\begin{equation}
    \bm{f}_{00} \cdot \Tilde{\bm{\mu}} \leq l.
\end{equation}
This assesses if $\Tilde{\bm{\mu}}$ is on the side of the hyperplane, given by $\bm{f}_{00}$, such that it could be inside the tetrahedron.

The above procedure is then repeated for each combination of three vectors in $\Tilde{\mathfrak{L}}$ giving $4$ inequalities characterised by the set of $4$ vectors $\mathfrak{f} = \{ \bm{f}_{00}, \bm{f}_{01}, \bm{f}_{10}, \bm{f}_{11} \}$. Using Table.~\ref{montones} for the values of the dot product, the complete set of monotones are 
\begin{equation}
    \{ \bm{f}_{00} \cdot \Tilde{\bm{\mu}} \leq l, \bm{f}_{01} \cdot \Tilde{\bm{\mu}} \leq m, \bm{f}_{10} \cdot \Tilde{\bm{\mu}} \leq o, \bm{f}_{11} \cdot \Tilde{\bm{\mu}} \leq q \}
\end{equation}
These 4 inequalities are a complete set of monotones such that if all four inequalities are satisfied if and only if there exists an allowed operation mapping $\mathcal{N} \rightarrow \mathcal{M}$. 
\def\arraystretch{2}
\begin{table}[h]
\begin{tabular}{c|c|c|c|c|}
\cline{2-5}
\multicolumn{1}{l|}{}                       & $\bm{f}_{00}$ & $\bm{f}_{01}$ & $\bm{f}_{10}$ & $\bm{f}_{11}$ \\ \hline
\multicolumn{1}{|c|}{$\Tilde{\bm{l}}_{00}$} & $l$           & $m$           & $o$           & $\delta$      \\ \hline
\multicolumn{1}{|c|}{$\Tilde{\bm{l}}_{01}$} & $l$           & $m$           & $\gamma$      & $q$           \\ \hline
\multicolumn{1}{|c|}{$\Tilde{\bm{l}}_{10}$} & $l$           & $\beta$       & $o$           & $q$           \\ \hline
\multicolumn{1}{|c|}{$\Tilde{\bm{l}}_{11}$} & $\alpha$      & $m$           & $o$           & $q$           \\ \hline
\end{tabular}
\caption{The table gives the values of the dot product of the elements of $\Tilde{\mathfrak{L}}$, given by $\Tilde{\bm{l}}_{ij}$, and the elements of $\mathfrak{f}$, given by $\bm{f}_{ij}$. The set of values $\{ \alpha, \beta, \gamma, \delta \}$ are constants unrelated to each other. \label{montones}}
\end{table}
\section{\label{sec:AppendixFour}Supplementary Material D: n Dimensional Channels}
Result~\ref{result3} characterises the transformation condition of when considering Weyl-covariant channels acting on $d_{S} = n$ dimensional system. In this case $\dynamicRT$ becomes $\dynamicRT_{n}$. The result states that an allowed operation exists between a channel $\mathcal{N}$ and a channel $\mathcal{M}$ if a vector of eigenvalues of the Choi-state of $\mathcal{M}$ is in the convex hull of local cyclic permutations of a vector of eigenvalues of the Choi-state of $\mathcal{N}$. We provide a full derivation here after an quick review of Weyl-covariant channels. For more details on Weyl-covariant channels see \cite{watrous_2018}.

A channel is Weyl-covariant if 
\begin{equation}
   W_{ab}\mathcal{N}(\rho)W_{ab}^{\dagger} = \mathcal{N}(W_{ab} \rho W_{ab}^{\dagger}) ~ \forall ~ a,b, \rho
\end{equation}
where $W_{ab}$ are the discrete Weyl operators, 
\begin{equation}
    W_{ab} = \sum_{c} \Omega^{bc} \ket{c + a ~ \textrm{mod} ~d_{S}}\bra{c}, \hspace{2mm} \Omega = e^{\frac{2 \pi i}{d_{S}}}, \hspace{2mm} a,b \in \{0, 1, ~... ~,d_{S}-1\}.
\end{equation}
These discrete Weyl operators can be used to generate the Bell basis states in $d_{S}^{2}$ dimensional space, as mentioned in Supplementary Material B. They are transformed in the following ways by transposition, the adjoint and multiplication \cite{watrous_2018},
\begin{equation}
    \begin{split}
        & (W_{\alpha, \beta}) ^ {t} = \Omega^{-\alpha \beta}W_{-\alpha, \beta}, \\
        & (W_{\alpha, \beta}) ^ {\dagger} = \Omega^{\alpha \beta}W_{-\alpha, -\beta}, \\
        & W_{\alpha, \beta}W_{\gamma, \delta} = \Omega^{\beta\gamma}W_{\alpha + \gamma, \beta + \delta} = \Omega^{\beta\gamma - \alpha\delta}W_{\gamma, \delta}W_{\alpha, \beta}. \label{weylIdentites}
    \end{split}
\end{equation}
These identities are employed in the subsequent proofs.

\subsection{Proof of Statement 1 implies statement 2 (Result~\ref{result3})}
We begin by showing that statement 1 of Result~\ref{result3} implies statement 2. We then show that statement 2 implies statement 1. 
\begin{proof}
To begin, it is noted that Weyl-covariant channels are a subset of mixed unitary channels and can always be written in the following form
\begin{equation}
    \mathcal{N}(\rho) = \sum_{ab} p_{ab} W_{ab} \rho W^{\dagger}_{ab},
\end{equation}
where $p_{ab}$ are elements of a probability vector \cite{watrous_2018}. 

When a restriction is made such that the allowed operations of the underlying static resource theory of informational non-equilibrium are Weyl-covariant channels, $\freeS = \mathfrak{U}_{\mathcal{W}}$, the allowed operations of the resource preservability theory are
\begin{equation}
    \begin{split}
        \choi{M} &= (\mathcal{P} \otimes \mathcal{E} ) \big( \choi{N} \big), \label{Choi State allowed operations Weyl}
    \end{split}
\end{equation}
where $\mathcal{P,E} \in \mathfrak{U}_{\mathcal{W}}$ and $\choi{N}, \choi{M}$ are Choi-states of Weyl-covariant channels. It can be shown that $\choi{M}$ remains the Choi-state of a Weyl-covariant channel if $\mathcal{N,P,E} \in \mathfrak{U}_{\mathcal{W}}$. (We ignore the shared randomness above and in the following proof for clarity, but the same result is reached if it is included). 

The form of the allowed operations can be simplified under the observation that the Choi-states of a Weyl-covariant channel are diagonal in the Bell basis,
\begin{equation}
    \begin{split}
        \choi{N} &= (\mathcal{I} \otimes \mathcal{N}) (\ketbra{\Phi_{00}}) \\
        &= (\mathcal{I} \otimes \sum_{ab} p_{ab} W_{ab}(\cdot)W^{\dagger}_{ab}) (\ketbra{\Phi_{00}}) \\
        &= \sum_{ab} p_{ab} (\mathcal{I} \otimes W_{ab})(\ketbra{\Phi_{00}}) (\mathcal{I} \otimes W^{\dagger}_{ab}) \\
        &= \sum_{ab} p_{ab} \ketbra{\Phi_{ab}},
    \end{split}
\end{equation}
where we have used the fact that discrete Weyl operators can be used to generate the Bell basis states, $\ket{\Phi_{ab}} = (\mathbb{I} \otimes W_{ab}) \ket{\Phi_{00}}$. This is the reverse of what was shown in Supplementary Material A, that Choi-states that are diagonal in the Bell basis correspond to Weyl-covariant channels.

As in the qubit case, the action of channels can therefore be captured through the eigenvalues of their Choi-state.   

By introduction two general Weyl-covariant channels, 
\begin{equation}
    \mathcal{P}(\rho) = \sum_{cd} p_{cd} W_{cd}\rho W_{cd}^{\dagger}, \hspace{5mm}  \mathcal{E}(\rho) = \sum_{ef} q_{ef} W_{ef}\rho W_{ef}^{\dagger},
\end{equation}
the allowed operations can be rewritten as 
\begin{equation}
    \bm{\mu} = \sum_{cd} \sum_{ef} p_{cd} q_{ef} \mathbb{B}^{(cd,ef)} \bm{\lambda},
\end{equation}
where $\bm{\lambda}, \bm{\mu} \in \mathbb{R}^{d_{S}^{2}}$ are vectors of eigenvalues of the Choi-states $\choi{N}, \choi{M}$ respectively and $\mathbb{B}^{(cd,ef)}$ is the matrix with elements $B^{cd,ef}_{nm, kl} = \abs{ \bra{\Phi_{kl}} W_{cd} \otimes W_{ef} \ket{\Phi_{nm}} }^{2}$.

Mirroring the qubit case, to complete the proof we need to show that $\mathbb{B}^{(cd,ef)}$ is a permutation from the set of local cyclic permutations.  

The matrix elements of $\mathbb{B}^{(cd,ef)}$ are given by
\begin{equation}
    \begin{split}
        B_{nm, kl}^{cd,ef} &= | \bra{\Phi_{kl}} \mathbb{I} \otimes W_{ef}(W_{cd})^{t} \ket{\Phi_{nm}} |^{2} \\
        &=  | \bra{\Phi_{kl}} \mathbb{I} \otimes \Omega^{-cd} W_{ef}W_{-c,d} \ket{\Phi_{nm}} |^{2} \\
        &= | \bra{\Phi_{kl}} \mathbb{I} \otimes \Omega^{-cd-cf} W_{e-c,d+f} \ket{\Phi_{nm}} |^{2} \\
        &= | \bra{\Phi_{kl}} \ket{\Phi_{n+e-c, m+d+f}} |^{2}, \\
        &= \delta_{k, n+e-c} ~ \delta_{l, m+d+f}. 
    \end{split}
\end{equation}
Therefore, the matrix $\mathbb{B}^{(cd,ef)}$ can be written as 
\begin{equation}
  \begin{split}
        \mathbb{B}^{(cd,ef)} &= \sum_{kl,nm} \delta_{k, n+e-c} ~\delta_{l, m+d+f} \ket{k,l}\bra{n,m} \\
        &= \sum_{kl,nm} \delta_{k, n+e-c} \ket{k}\bra{n} \otimes \delta_{l, m+d+f} \ket{l}\bra{m}, \\
        &= \sum_{n,m} \ket{n+e-c ~ \textrm{mod} ~ d_{S}} \bra{n} \otimes \ket{m+d+f ~ \textrm{mod} ~ d_{S}}\bra{m} \\
        &= \sum_{n}  \ket{n+e-c ~ \textrm{mod} ~ d_{S}} \bra{n} \otimes \sum_{m} \ket{m+d+f ~ \textrm{mod} ~ d_{S}}\bra{m} \\
        &= W_{e-c,0} \otimes W_{d+f,0} \\
        &= X_{d_{S}}^{e-c} \otimes X_{d_{S}}^{d+f}, \label{LocalPermutationsProof}
  \end{split}
\end{equation}
where $X_{d_{S}}^{\kappa}$ is a generalised Pauli $X$ in $d_{S}$ dimensions. 

To see that $X_{d_{S}}^{e-c}$ is from the set of cyclic (or even) permutations we consider it in the form
\begin{equation}
    \begin{split}
        X_{d_{S}}^{e-c} &= \sum_{n} \ket{n+e-c ~ \textrm{mod} ~ d_{S}} \bra{n},  \\ \label{alocalCyclicPermutation}
        &=  \sum_{n} \ket{i(n)} \bra{n},
    \end{split}
\end{equation}
and use the example of qutrits, $d_{S} = 3.$ The index $n$ counts up in the range $0 \leq n < d_{S} - 1$ whilst $e-c$ is a constant for a given $\mathbb{B}^{(cd,ef)}$. The function $i(n)$ for different constants $e-c$ is evaluated for different inputs $n$ in Table~\ref{Table2:cyclicPermsTable}. From here it can be seen that $X_{d_{S}}^{e-c}$ permits only the cyclic permutations. In order for a non-cyclic permutation to occur $n$ would need to be negative. It can be seen that this is never the case. The same is true for $X_{d_{S}}^{d+f}$. Hence, $\mathbb{B}^{(cd,ef)}$ is the tensor product of two cyclic permutations and therefore from the set of local cyclic permutations.

an allowed operation therefore exists between $\mathcal{N} \rightarrow \mathcal{M}$ if there exists a matrix between $\bm{\lambda}$ and $\bm{\mu}$ that is a convex combination of local cyclic permutations. This completes the proof of this direction of Result~\ref{result3}.

\end{proof}

\def\arraystretch{1.3}
\begin{table}[h]
\begin{tabular}{ccccllll}
\cline{1-4}
\multicolumn{1}{|c|}{\textbf{$n$}}     & \multicolumn{1}{c|}{\textbf{0}} & \multicolumn{1}{c|}{\textbf{1}} & \multicolumn{1}{c|}{\textbf{2}} &  &  &  &  \\ \cline{1-4}
\multicolumn{1}{|c|}{\textbf{$i=n$}}   & \multicolumn{1}{c|}{0}          & \multicolumn{1}{c|}{1}          & \multicolumn{1}{c|}{2}          &  &  &  &  \\ \cline{1-4}
\multicolumn{1}{|c|}{\textbf{$i=n+1$}} & \multicolumn{1}{c|}{1}          & \multicolumn{1}{c|}{2}          & \multicolumn{1}{c|}{0}          &  &  &  &  \\ \cline{1-4}
\multicolumn{1}{|c|}{\textbf{$i=n+2$}} & \multicolumn{1}{c|}{2}          & \multicolumn{1}{c|}{0}          & \multicolumn{1}{c|}{1}          &  &  &  &  \\ \cline{1-4}
\multicolumn{1}{|c|}{\textbf{$i=n-1$}} & \multicolumn{1}{c|}{2}          & \multicolumn{1}{c|}{0}          & \multicolumn{1}{c|}{1}          &  &  &  &  \\ \cline{1-4}
\multicolumn{1}{|c|}{\textbf{$i=n-2$}} & \multicolumn{1}{c|}{1}          & \multicolumn{1}{c|}{2}          & \multicolumn{1}{c|}{0}          &  &  &  &  \\ \cline{1-4}
\multicolumn{1}{l}{}                   & \multicolumn{1}{l}{}            & \multicolumn{1}{l}{}            & \multicolumn{1}{l}{}            &  &  &  & 
\end{tabular}
\caption{A Table showing the function $i(n)=n+\nu ~ \textrm{mod} ~ d_{in}$ for various values of the constant $\nu = e-c$ when considering qutrit channels, such that $d_{in}=3$. \label{Table2:cyclicPermsTable} }
\end{table}

\subsection{Proof of Statement 2 implies statement 1 (Result~\ref{result3})}

\begin{proof} Assuming statement 2 holds, we can write 
\begin{equation}
    \bm{\mu} = \sum_{\alpha, \beta} p_{\alpha \beta}  X_{d_{S}}^{\alpha} \otimes X_{d_{S}}^{\beta} \bm{\lambda},
\end{equation}
where we have introduced $\alpha=e-c$ and $\beta=d+f$ for clarity. Translating back to the Choi-states of the channels these vectors represent, we find 
\begin{equation}
    \choi{M} = \sum_{nm} p_{nm} (X_{d_{S}}^{\alpha} \otimes X_{d_{S}}^{\beta}) (\choi{N}) (X_{d_{S}}^{\alpha} \otimes X_{d_{S}}^{\beta}),
\end{equation}
where $\choi{M}$ and $\choi{N}$ are Choi-states of the Weyl covariant channels $\mathcal{M}$ and $\mathcal{N}$ -- and hence diagonal in the Bell basis -- with eigenvalues given by the elements of $\bm{\mu}$ and $\bm{\lambda}$ respectively. This can be re-written as 
\begin{equation}
    \begin{split}
        \choi{M} &= \sum_{\alpha \beta} p_{\alpha \beta} (X_{d_{S}}^{\alpha} \otimes X_{d_{S}}^{\beta}) (\mathcal{I} \otimes \mathcal{N})(\ketbra{\Phi_{00}}) (X_{d_{S}}^{\alpha} \otimes X_{d_{S}}^{\beta}), \\
        &= \sum_{\alpha \beta} p_{\alpha \beta} (\mathcal{I} \otimes X_{d_{S}}^{\beta}) (\mathcal{I} \otimes \mathcal{N})(\mathcal{I} \otimes (X_{d_{S}}^{\alpha})^{t})(\ketbra{\Phi_{00}})(\mathcal{I} \otimes (X_{d_{S}}^{\alpha})^{t}) (\mathcal{I} \otimes X_{d_{S}}^{\beta}) , \\
        &= \sum_{nm} p_{nm} (\mathcal{I} \otimes X_{d_{S}}^{\beta} [\cdot] \circ \mathcal{N} \circ (X_{d_{S}}^{\alpha})^{t}[\cdot])(\ketbra{\Phi_{00}}),
    \end{split}
\end{equation}
where $X_{d_{S}}^{\beta} [\cdot]$ and $(X_{d_{S}}^{\alpha})^{t}[\cdot]$ are unitary channels that apply $X_{d_{S}}^{\beta}$ and $(X_{d_{S}}^{\alpha})^{t}$ respectively. Hence, from the Choi-Jamiołkowski isomorphism,
\begin{equation}
    \begin{split}
        \mathcal{M}(\rho) &= \sum_{nm} p_{nm} X_{d_{S}}^{\beta} [\cdot] \circ \mathcal{N} \circ (X_{d_{S}}^{\alpha})^{t}[\cdot] (\rho), \\
        &= \sum_{nm} p_{nm} \mathcal{N} \circ X_{d_{S}}^{\beta} [\cdot] \circ (X_{d_{S}}^{\alpha})^{t}[\cdot] (\rho), \\
        &= \mathcal{N} \circ \mathcal{P},
    \end{split}
\end{equation}
where $\mathcal{P}(\rho) = \sum_{nm} p_{nm} X_{d_{S}}^{\beta} [\cdot] \circ (X_{d_{S}}^{\alpha})^{t}[\cdot] (\rho)$ which is a pre-processing Weyl covariant channel. In the penultimate line we have used the fact that $\mathcal{N}$ is a Weyl-covariant channel and that $X_{d_{S}}^{\beta}$ is a discrete Weyl operator for all $\beta$. Therefore, if statement 2 holds then there always exists a Weyl-covariant pre-processing channel that achieves that mapping $\mathcal{N} \rightarrow \mathcal{M}$. Hence, statement 2 implies statement 1. This completes the proof of Result \ref{result3}.
\end{proof}

\section{\label{Appendix five}Supplementary Material E: Operational Interpretation}
Result~\ref{result5} gives an operational interpretation of the pre-order established on $\freeS$ by $\dynamicRT$ and $\dynamicRT_{n}$. The existence of an allowed operation $\Pi \in \freeD$, such that $\mathcal{M}=\Pi(\mathcal{N})$, where $\mathcal{M,N} \in \freeS$ can be established by acting the channels on one half of a set of bipartite states and then comparing the expectation values of a Bell state measurement on the output states. 

We give a derivation of this result here and detail some other available strategies then the one given in the main text. 

Firstly, we recall that within the context of the existence of an allowed operations in $\dynamicRT$, the Choi-states of qubit unital channels and Weyl-covariant channels with finite dimension can always be considered in the Bell basis such that they are diagonal. Hence, the Choi-states can be mapped to a vector of eigenvalues, 
\begin{equation}
    \choi{N} = (\mathcal{I} \otimes \mathcal{N}) (\ketbra{\Phi_{00}}) \rightarrow \bm{\lambda} \in \mathbb{R}^{d_{S}^{2}}. 
\end{equation}
Due to normalisation the final dimension of the vector is redundant and can be dropped without loss. This allows $\choi{N}$ to be mapped to a vector in $\mathbb{R}^{d_{S}^{2}-1}$. The complete mapping is then 
\begin{equation}
    \choi{N}_{00} \rightarrow \bm{l_{00}} \in \mathfrak{L} \rightarrow \Tilde{\bm{l}}_{00} \in \Tilde{\mathfrak{L}},
\end{equation}
where $\mathfrak{L} / \Tilde{\mathfrak{L}}$ are as defined above (with $\bm{l}_{00}=\bm{\lambda}$) and the subscript indicates that the Bell basis state $\ket{\Phi_{00}}$ was used in the Choi-state definition. This is an alternative way to see the generation of the set of vectors $\Tilde{\mathfrak{L}}$ that define the simplex in $\mathbb{R}^{d_{S}^{2}-1}$. The other elements of the set $\Tilde{\mathfrak{L}}$ can be generated by using different Bell states in the definition of the Choi-state, and following the same mapping. For example,
\begin{equation}
    \begin{split}
        \choi{N}_{01} &= (\mathcal{I} \otimes \mathcal{N}) (\ketbra{\Phi_{01}}) \rightarrow \bm{l_{01}} \in \mathfrak{L} \rightarrow \Tilde{\bm{l}}_{01} \in \Tilde{\mathfrak{L}}, \\
        \choi{N}_{10} &= (\mathcal{I} \otimes \mathcal{N}) (\ketbra{\Phi_{10}}) \rightarrow \bm{l_{10}} \in \mathfrak{L} \rightarrow \Tilde{\bm{l}}_{10} \in \Tilde{\mathfrak{L}},    
    \end{split}
\end{equation}
and so on. Using a different Bell basis state leaves the eigenvalues invariant but permutes their position. Therefore, each element of $\Tilde{\mathfrak{L}}$ can be considered to be a mapping of the Choi-state of a channel $\mathcal{N} \in \freeS$ to a vector in $\mathbb{R}^{d_{S}^{2}-1}$ where a different Bell state has been used in the Choi-state definition. Note that previously $\Tilde{\mathfrak{L}}$ was generated from $\bm{\lambda}$ by applying the local permutations to $\bm{\lambda}$, $\mathfrak{L} = \{ \perm{i} \bm{\lambda} : \forall ~ \perm{i} \in \mathfrak{P} \}$ and then dropping the final redundant dimension.

From this, it can be observed that 
\begin{equation}
    \sum_{ij} \choi{N}_{ij} = \mathbb{I},
\end{equation}
meaning that the states $\choi{N}_{ij}$ can be considered elements of a quantum measurement, which is mathematically described by the so-called \textit{positive operator-valued measure} (POVM) \cite{nielsen_chuang_2010}. This alternative approach to generating the vertices of the simplex and the above POVM will be employed in the subsequent proofs. 

Before returning to the operational interpretation, the following Lemma is proved.
\begin{lemma} \label{lemma5.1}
Let $\mathcal{M,N}$ be two arbitrary channels (not necessarily unital).
Then the following two statements are equivalent for $\freeD$ induced by $\freeS$ if $\freeS$ is the set of unital channels:
\begin{enumerate}
\item $\mathcal{M}=\Pi(\mathcal{N})$ for some $\Pi \in \freeD$.
\item $\mathcal{M}^{\dagger}=\Theta(\mathcal{N}^{\dagger})$ for some $\Theta \in \freeD$.
\end{enumerate}
\end{lemma}
\begin{proof}
Given the set of mixed unitary channels and the set of Weyl-covariant channels are subsets of the set of unital channels, this proof holds if $\freeS = \mathcal{U}_{M}$ or $\freeS = \mathcal{U}_{\mathcal{W}}$ also. 

Suppose that $\mathcal{M}=\Pi(\mathcal{N})$ for some $\Pi \in \freeD$. This can be written as
\begin{equation}
    \mathcal{M} = \sum_{\kappa} p_{\kappa} \mathcal{E}_\kappa \circ \mathcal{N} \circ \mathcal{P}_\kappa,
\end{equation}
where $\mathcal{P}_{\kappa}, \mathcal{E}_{\kappa} \in \freeS$. The adjoint of $\mathcal{M}$ is then 
\begin{equation}
    \mathcal{M}^{\dagger} = \sum_{\kappa} p_{\kappa} \mathcal{P}_{\kappa}^{\dagger} \circ \mathcal{N}^{\dagger} \circ \mathcal{E}^{\dagger}_{\kappa},
\end{equation}
where $\mathcal{P}^{\dagger}_{\kappa}, \mathcal{E}^{\dagger}_{\kappa} \in \freeS$. Hence, if an allowed operations exists between $\mathcal{N}$ and $\mathcal{M}$ then an allowed operations exists between $\mathcal{N}^{\dagger}$ and $\mathcal{M}^{\dagger}$. Likewise, if an allowed operations exists between $\mathcal{N}^{\dagger}$ and $\mathcal{M}^{\dagger}$, then the same proof shows that an allowed operations exists between $\mathcal{N}$ and $\mathcal{M}$. 

\end{proof}

\subsection{Proof of Result 5}

\begin{proof}
Returning to the operational interpretation, we note that the set of $d_{s}^{2}$ many states $\{ \rho_{ij} \}$ arise from the geometrical interpretation of the allowed operations through vectors in $\mathbb{R}^{d_{S}^{2}}$ space. Hence, the linear independence of the vectors in $\Tilde{\mathfrak{L}}$ is once again important. We assume they are linear independent in this section ($\textrm{dimSpan} (\mathfrak{K}) = d_{S}^{2} - 1$) but a similar methodology could be applied if they were not.

The hyper-planes that encompass each facet of the simplex are defined by a set of vectors $\mathfrak{f} = \{\bm{f}_{ij}: \bm{f}_{ij} \in \mathbb{R}^{d_{S}^{2} -1} ~ \forall ~ i,j \}$, from which a complete set of monotones can be found as seen in Supplementary Material C. From $\mathfrak{f}$, a set of vectors $\mathfrak{Q} = \{ \bm{Q}_{ij} : \bm{Q}_{ij} \in \mathbb{R}^{d_{S}^{2}} ~ \forall ~ i,j \}$ can be created, the elements of which can be considered to be the diagonal entries of a bipartite state in the Bell basis.  

Each vector $\bm{Q}_{ij}$ is generated from an element of $\mathfrak{f}$ as
\begin{equation}
    \bm{Q}_{ij} = \frac{\bm{F}_{ij} + r_{ij}\bm{\Lambda}}{w_{ij}},
\end{equation}
where $ \bm{F}_{ij} \in \mathbb{R}^{d_{S}^{2}}$ is a higher dimensional embedding of $\bm{f}_{ij}$, such that a zero has been appended to the vector in the added dimension (i.e $\bm{F}^{t}_{ij} = [\bm{f}_{ij}, 0])$, $r_{ij}$ is a number such that all elements of $\bm{Q}_{ij}$ are positive, $\bm{\Lambda}$ is the vector representing the state preparation channel of the maximally mixed state and $w_{ij}$ is a normalisation. 

Each $\bm{Q}_{ij}$ can then be taken to be the diagonal elements of some state in the Bell basis, $\sigma_{ij}$. This will be a valid density operator as it will have trace one due to the normalisation and be positive due to $r_{ij}$ ensuring all entries are positive.   

Given all Choi-states can be considered to be diagonal in the Bell basis, 
\begin{equation}
    \begin{split}
        \textrm{Tr}(\choi{N}_{ij}\sigma_{ij}) &= \bm{l}_{ij} \cdot \bm{Q}_{ij}, \\
        &= \frac{\bm{l}_{ij} \cdot \bm{F}_{ij} + r_{ij} ~ \bm{l}_{ij} \cdot \bm{\Lambda}}{w_{ij}}, \\
        &= \frac{\Tilde{\bm{l}}_{ij} \cdot \bm{f}_{ij} + r_{ij}/4}{w_{ij}}. \label{operationalInterpretation}
    \end{split}
\end{equation}
In the above, the first term simplifies due to the zero in the final dimension of $\bm{F}_{ij}$ and the second due to the elements of $\bm{l}_{ij}$ summing to one. 

The value of $\Tilde{\bm{l}}_{ij} \cdot \bm{f}_{ij}$ depends on the geometry in $\mathbb{R}^{d_{S}^{2}-1}$ and can be considered as a distance like measure from the origin to one facet of the simplex, as seen in Supplementary Material C, and is therefore a constant.

Assuming $\Tilde{\bm{l}}_{ij} \cdot \bm{f}_{ij} = l$, and now replacing $\choi{N}_{ij}$ with $\choi{M}_{ij}$, 
\begin{equation}
    \begin{split}
        \textrm{Tr}(\choi{M}_{ij}\sigma_{ij}) &= \frac{\Tilde{\bm{\mu}}_{ij} \cdot \bm{f}_{ij} + r_{ij}/4}{w_{ij}},
    \end{split}
\end{equation}
where $\Tilde{\bm{\mu}}_{ij} \in \mathbb{R}^{d_{S}^{2}-1}$. If $\Tilde{\bm{\mu}}$ falls on the correct side of the hyper-plane $\bm{f}_{ij}$ encompassing the face of the simplex then $\Tilde{\bm{\mu}}_{ij} \cdot \bm{f}_{ij} < l$, and hence,
\begin{equation}
    \begin{split}
        \textrm{Tr}(\choi{M}_{ij}\sigma_{ij}) &< \frac{l + r_{ij}/4}{w_{ij}} \\
        &= \textrm{Tr}(\choi{N}_{ij}\sigma_{ij})
    \end{split}
\end{equation}
In order to evaluate if $\bm{\mu}$ is on the correct side of each face of the simplex, the above must be true for all $i,j$ (see Table.\ref{montones} for the qubit case). 

The indices in the above equations need not be the same to reach the same result. It was shown above that $\Tilde{\bm{l}}_{\alpha, \beta} \cdot \bm{f}_{ij}$ equals the same constant for multiple different values of $(\alpha, \beta)$. Hence, for these values of $(\alpha, \beta)$, $\textrm{Tr}(\choi{N}_{\alpha \beta}\sigma_{ij})$
will give the same value as in Eq.~\eqref{operationalInterpretation}. There is therefore a few ways to build this task from the geometrical interpretation, and these are made explicitly clear in the qubit example below. 

To complete the proof, $\textrm{Tr}(\choi{N}_{ij}\sigma_{ij})$ can be rewritten, but a distinction between the $d_{S}=2$ and $d_{S}=n$ case must first be made. In the qubit channel case, $d_{S}=2$, the Choi-states of mixed unitary channels can be diagonalised in the Bell basis under local unitaries. Given all unitaries are in $\freeD$, we consider only diagonal Choi-states when discussing the existence of allowed operations. However, these diagonal Choi-states really represent a Weyl-covariant channel, $\widetilde{\mathcal{N}}$, that is equivalent to a channel $\mathcal{N}$ up to a prepossessing and post processing unitary, 
\begin{equation}
    \widetilde{\mathcal{N}}(\rho) = V \mathcal{N}(U\rho U^{\dagger})V^{\dagger}. 
\end{equation}
When unpacking the Choi-state there is a need to be explicit about what channel is being referred to,
\begin{equation}
    \choi{\Tilde{N}}_{ij} = (U_{1} \otimes U_{2}) \choi{N}_{ij} (U_{1} \otimes U_{2})^{\dagger},
\end{equation}
where $\choi{\Tilde{N}}_{ij}$ is diagonal in the Bell basis and the Choi-state of a Weyl-covariant channel, and $U_{1}/U_{2}$ are unitaries. Note that $U_{1}$, $U_{2}$ are dependent on $i,j$ but this is kept implicit for simplicity

Returning to the left hand side of Eq.~\eqref{operationalInterpretation} with this explicit notation,
\begin{equation}
    \begin{split}
        \textrm{Tr}(\choi{\Tilde{N}}_{ij}\sigma_{ij}) &= \textrm{Tr} \big((U_{1} \otimes U_{2}) \choi{N}_{ij} (U_{1} \otimes U_{2})^{\dagger} \sigma_{ij} \big) \\
        &=  \textrm{Tr} \big( (U_{1} \otimes U_{2})^{\dagger} \sigma_{ij}(U_{1} \otimes U_{2}) \choi{N}_{ij} \big) \\
        &= \textrm{Tr} \big(\tau_{ij} \choi{N}_{ij} \big) \\
        &= \textrm{Tr} \big((\tau_{ij}) (\mathcal{I} \otimes \mathcal{N}) \ketbra{\Phi_{ij}} \big)  \\
        &= \bra{\Phi_{ij}} (\mathcal{I} \otimes \mathcal{N}^{\dagger})(\tau_{ij}) \ket{\Phi_{ij}},
    \end{split}
\end{equation}
where $\tau_{ij} = (U_{1} \otimes U_{2})^{\dagger} \sigma_{ij}(U_{1} \otimes U_{2})$. Therefore, there exists an allowed operation converting $\mathcal{N}$ into $\mathcal{M}$ and only if 
\begin{equation}
    \bra{\Phi_{ij}} (\mathcal{I} \otimes \mathcal{N}^{\dagger})(\tau_{ij}) \ket{\Phi_{ij}} \geq \bra{\Phi_{ij}} (\mathcal{I} \otimes \mathcal{M}^{\dagger})(\tau_{ij}) \ket{\Phi_{ij}} ~ \forall~i,j
\end{equation}
for a set $d_{s}^{2}$ many quantum states $\{ \tau_{ij} \}$, each coming from an element of $\mathfrak{f}$. This is true if and only if there also exists an allowed operation between $\mathcal{N}^{\dagger}$ and $\mathcal{M}^{\dagger}$ as shown by Lemma~\ref{lemma5.1}. By repeating the above process but starting with channels $\mathcal{N}^{\dagger}$ and $\mathcal{M}^{\dagger}$, we conclude the following desired result:
\\
\\
\textit{Let $\mathcal{N}\in\freeS$ be given. Then there exists a set of $d_S^2$ many quantum states $\{\rho_{ij}\}$ which makes the following two statement equivalent:
\begin{enumerate}
\item There exists a $\Pi \in \freeD$ such that $\mathcal{M} = \Pi(\mathcal{N})$.
\item The following $d_S^2$ many inequalities are satisfied:
\begin{equation}
    \bra{\Phi_{ij}} (\mathcal{I} \otimes \mathcal{N})(\rho_{ij}) \ket{\Phi_{ij}} \geq \bra{\Phi_{ij}} (\mathcal{I} \otimes \mathcal{M})(\rho_{ij}) \ket{\Phi_{ij}} ~ \forall~i,j,
\end{equation}
\end{enumerate}
}

For channels where $d_{S}>2$ we consider $\freeS = \mathfrak{U}_{\mathcal{W}}$. The Choi-state of channels in $\mathfrak{U}_{\mathcal{W}}$ are already diagonal in the Bell basis and hence no diagonalisation needs to be `undone'. The rest of the above logic holds giving the same result as the qubit case. This completes the proof.
\end{proof}
\subsection{Qubit Channels Example}
For clarity, an example of the formation of a set of states for the operational interpretation when considering qubit channels, $d_{S} =  2$, is included. A comparison of Bell state measurements on mutiple copies of an ensemble of bipartite states, after the application of the channels $\mathcal{N}$ and $\mathcal{M}$ on a subspace of the states, will confirm the existence of an allowed operation $\Pi \in \freeD$, such that $\mathcal{M} = \Pi(\mathcal{N})$, where $\mathcal{M,N} \in \freeS$. We assume $\textrm{dimSpan} (\mathfrak{K}) = 3$ in the following.

The states are constructed from the monotones of the channel $\mathcal{N}^{\dagger}$. The existence of an allowed operation between the channels $\mathcal{N}^{\dagger}$ and $\mathcal{M}^{\dagger}$ implies the existence of an allowed operation between $\mathcal{N}$ and $\mathcal{M}$, as seen above (Lemma~\ref{lemma5.1}). Proving the existence of an allowed operation between the adjoint channels gives an operational interpretation in terms of $\mathcal{N}$ and $\mathcal{M}$, rather then $\mathcal{N}^{\dagger}$ and $\mathcal{M}^{\dagger}$. 

In this section, the vector $\bm{\lambda} \in \mathbb{R}^{4}$ is a vector of eigenvalues of the Choi-state $\choi{N^{\dagger}}$, $\bm{\lambda}^{t} = [\lambda_{00}, \lambda_{01}, \lambda_{10}, \lambda_{11}]$, and $\bm{\mu} \in \mathbb{R}^{4}$ is a vector of eigenvalues of the Choi-state $\choi{M^{\dagger}}$, $ \bm{\mu}^{t} = [\mu_{00}, \mu_{01}, \mu_{10}, \mu_{11}]$. 

Fig.~\ref{fig3:operationalInterpretation} highlights one face of the simplex given by the set $\Tilde{\mathfrak{L}}$ (given in full in Supplementary Material C, Eq.~\eqref{setTildeL}). The plane that encompasses this face is defined by the vector 
\begin{equation}
    \bm{f}_{00} = \begin{bmatrix}
           a \\
           b \\
            c 
         \end{bmatrix} \in \mathfrak{f},
\end{equation}
where $a,b,c \in \mathbb{R}$. Recall from Supplementary Material C that 
\begin{equation}
    \bm{f}_{00} \cdot \Tilde{\bm{l}}_{00} = \bm{f}_{00} \cdot \Tilde{\bm{l}}_{01} = \bm{f}_{00} \cdot \Tilde{\bm{l}}_{10} = l, \label{distanceToSimplex}
\end{equation}
and that if a channel giving by $\Tilde{\bm{\mu}}$ is on the correct side of the plane $\bm{f}_{00}$ (such that it is inside the simplex) then $\bm{f}_{00} \cdot \Tilde{\bm{\mu}} < l.$

To create the set of states the vector $\bm{f}_{00}$ is first embedded in a higher dimensional space, 
\begin{equation}
     \bm{F}_{00} = \begin{bmatrix}
           a \\
           b \\
            c \\
            0
         \end{bmatrix}.
\end{equation}
The maximally mixed state, scaled by a number $r_{00}$, is then added to $\bm{F}_{00}$ to generate a vector with all positive elements. This state is then normalised, giving $\bm{Q}_{00}$ as 
\begin{equation}
    \bm{Q}_{00} = \frac{1}{w_{00}}\begin{bmatrix}
           a + r_{00}/4 \\
           b + r_{00}/4 \\
            c + r_{00}/4 \\
            r_{00}/4
         \end{bmatrix},
\end{equation}
where $w_{00}$ is the normalisation. 

The elements of $\bm{Q}_{00}$ are taken to be the diagonal elements of a density operator written in the Bell basis, 
\begin{equation}
    \sigma_{00} = \frac{1}{w_{00}} \begin{pmatrix}
       a + r_{00}/4 & 0 & 0 & 0 \\
        0 & b + r_{00}/4 & 0 & 0 \\
        0 & 0 & c + r_{00}/4 & 0 \\
        0 & 0 & 0 & r_{00}/4
    \end{pmatrix},
\end{equation}
where the off diagonal elements can be set to zero without loss of generality. Given $\choi{N}_{00}$ is diagonal in the Bell basis, 
\begin{equation}
    \begin{split}
        \textrm{Tr}(\choi{N}_{00}\sigma_{00}) &= \bm{l}_{00} \cdot \bm{Q}_{00}, \\
        &= \frac{[a\lambda_{00} + b\lambda_{01} + c\lambda_{10}] + r_{00}/4 ~ [\lambda_{00} + \lambda_{01} + \lambda_{10} + \lambda_{00}]}{w_{00}} \\
        &= \frac{\Tilde{\bm{l}}_{00} \cdot \bm{f}_{00} + r_{00}/4}{w_{00}} \\
        &= \frac{l + r_{00}/4}{w_{00}},
    \end{split}
\end{equation}
which is a constant that depends on the geometry of the problem. Hence, $\textrm{Tr}(\choi{N}_{00}\sigma_{00}) = p_{l}$ where $p_{l}$ is a constant that depends on $l$. Note that from Eq.~\eqref{distanceToSimplex} it can be seen that  
\begin{equation}
    \textrm{Tr}(\choi{N^{\dagger}}_{00}\sigma_{00}) = \textrm{Tr}(\choi{N^{\dagger}}_{01}\sigma_{00}) = \textrm{Tr}(\choi{N^{\dagger}}_{10}\sigma_{00}) = p_{l}. \label{eqaulExpectationValues}
\end{equation}
Moreover, it was stated above that the Choi-states produced using a different Bell basis state form the elements of a POVM. Hence, it can be seen that  
\begin{equation}
    \begin{split}
        \textrm{Tr}(\choi{N^{\dagger}}_{11}\sigma_{00}) &= 1 - \Bigg[  \textrm{Tr}(\choi{N^{\dagger}}_{00}\sigma_{00}) + \textrm{Tr}(\choi{N^{\dagger}}_{01}\sigma_{00}) + \textrm{Tr}(\choi{N^{\dagger}}_{10}\sigma_{00}) \Bigg] \\
        &= 1 - 3 p_{l}.
    \end{split}
\end{equation}
This will provide alternative strategies for the operational interpretation, detailed below. 

Now considering $\choi{M^{\dagger}}_{00}$ and $\sigma_{00}$, 
\begin{equation}
     \begin{split}
        \textrm{Tr}(\choi{M^{\dagger}}_{00}\sigma_{00}) &= \bm{\mu} \cdot \bm{Q}_{00}, \\
        &= \frac{[a\mu_{00} + b\mu_{01} + c\mu_{10}] + r_{00}/4 ~ [\mu_{00} + \mu_{01} + \mu_{10} + \mu_{00}]}{w_{00}} \\
        &= \frac{\bm{f}_{00} \cdot \Tilde{\bm{\mu}} + r_{00}/4}{w_{00}} \\
        & \leq \frac{l + r_{00}/4}{w_{00}},
    \end{split}
\end{equation}
where on the final line it is assumed that $\bm{\mu}$ is on the correct side of $\bm{f}_{00}$. If $\bm{\mu}$ is not on the correct side of $\bm{f}_{00}$ then 
\begin{equation}
     \textrm{Tr}(\choi{M^{\dagger}}_{00}\sigma_{00}) > \frac{l + r_{00}/4}{w_{00}}.
\end{equation}
For each plane in $\mathfrak{f}$ the above process can be repeated, creating a set of 4 states $\{ \sigma_{ij} \}_{i,j=0}^{1}$. Hence, we learn that there exists a free operation converting $\mathcal{N}$ into $\mathcal{M}$ (where Lemma~\ref{lemma5.1} has been used) if and only if
\begin{align}
\textrm{Tr}(\choi{M^{\dagger}}_{ij}\sigma_{ij})\le\textrm{Tr}(\choi{N^{\dagger}}_{ij}\sigma_{ij})\quad\forall\;i,j\in\{0,1\}.
\end{align}
Using the relation $\textrm{Tr}(\choi{M^{\dagger}}_{ij}\sigma_{ij})=\bra{\Phi_{ij}}(\mathcal{I}\otimes\mathcal{M})(\sigma_{ij})\ket{\Phi_{ij}}$ for every $i,j$ (and a similar relation also holds for $\mathcal{N}$), we conclude the following result: 
\\
\\
\textit{Let $\mathcal{N}\in\freeS$ be a given qubit channel. Then there exists a set of 4 quantum states $\{\rho_{ij}\}_{i,j=0}^1$ making the following two statement equivalent:
\begin{enumerate}
\item There exists a $\Pi \in \freeD$ such that $\mathcal{M} = \Pi(\mathcal{N})$.
\item The following four inequalities are satisfied:
\begin{equation}
   \begin{split}
        \bra{\Phi_{00}} (\mathcal{I} \otimes \mathcal{N}) (\rho_{00}) \ket{\Phi_{00}} \geq \bra{\Phi_{00}} (\mathcal{I} \otimes \mathcal{M}) (\rho_{00}) \ket{\Phi_{00}}, \\
         \bra{\Phi_{01}} (\mathcal{I} \otimes \mathcal{N}) (\rho_{01}) \ket{\Phi_{01}} \geq \bra{\Phi_{01}} (\mathcal{I} \otimes \mathcal{M}) (\rho_{01}) \ket{\Phi_{01}}, \\
          \bra{\Phi_{10}} (\mathcal{I} \otimes \mathcal{N}) (\rho_{10}) \ket{\Phi_{10}} \geq \bra{\Phi_{10}} (\mathcal{I} \otimes \mathcal{M}) (\rho_{10}) \ket{\Phi_{10}}, \\
           \bra{\Phi_{11}} (\mathcal{I} \otimes \mathcal{N}) (\rho_{11}) \ket{\Phi_{11}} \geq \bra{\Phi_{11}} (\mathcal{I} \otimes \mathcal{M}) (\rho_{11}) \ket{\Phi_{11}}.
   \end{split}
 \end{equation}
 \end{enumerate} }

These can be interpreted as applying the channel to one half of each state in the set and then measuring in the Bell basis. A comparison of the expectation values for the two different channels can then be made after multiple measurements to decided if an allowed operation exists. Each expectation value is an operational way of checking each face of the simplex to asses if $\Tilde{\bm{\mu}}$ sits inside or outside the simplex.  

We note that only the diagonal elements of the set of states $\{ \rho_{ij} \}^{4}$ need to be fixed. The freedom in the choice of off diagonal elements means it is always possible to make these states pure state.  

The above conditions have an operational interpretation as a state discrimination task. A referee distributes half of one of the states, $\rho_{ij}$, to Alice and the other half to Bob. Alice then sends her half of the state to Bob via the channel $\mathcal{N}$ or $\mathcal{M}$. Bob, now with the whole state, performs a measurement $\{ \ketbra{\Phi_{ij}} \}_{ij}$ -- a Bell state measurement. Bob succeeds at the game if for a state $\rho_{ij}$ he gets the POVM outcome associated to $\ket{\Phi_{ij}}$. If for a given state $\rho_{ij}$ in the set Bob succeeds with at least as high a success probability when Alice sends her half of the state with a channel $\mathcal{N}$ as when Alice send it with $\mathcal{M}$, then $\bm{\mu}$ sits on the correct side of one face of the simplex. If Bob succeeds with at least as high a success probability for all states in the set, then an allowed operation exists between $\mathcal{N}$ and $\mathcal{M}$. In other state discrimination games in the literature the average success probability is used as the figure of merit \cite{Bae_2015}. Interestingly, this game differs in that each individual probability of successful discrimination needs to be considered in order to probe the existence of an allowed operation.

From Eq.~\eqref{distanceToSimplex} it can be seen there are multiple ways of finding $l$. This leads to multiple operational methods of finding $p_{l}$, as seen in Eq.~\eqref{eqaulExpectationValues}. Therefore, there are multiple different strategies that could be employed in this operational interpretation. Table~\ref{Table3:operationalStratriges} shows the complete range of possible combinations that can be used to assess each face of the simplex. Four expectation values must be compared, each one in the set $\{ p_{l}, p_{m}, p_{o}, p_{q} \}$, with each checking one face of the simplex. Each must be compared only once. Hence, another equally valid strategy of assessing the existence of an allowed operation is 
\begin{equation}
   \begin{split}
        (p_{l} = ) \bra{\Phi_{00}} (\mathbb{I} \otimes \mathcal{N}) (\rho_{00}) \ket{\Phi_{00}} \geq \bra{\Phi_{00}} (\mathbb{I} \otimes \mathcal{M}) (\rho_{00}) \ket{\Phi_{00}}, \\
         (p_{m} = ) \bra{\Phi_{00}} (\mathbb{I} \otimes \mathcal{N}) (\rho_{01}) \ket{\Phi_{00}} \geq \bra{\Phi_{00}} (\mathbb{I} \otimes \mathcal{M}) (\rho_{01}) \ket{\Phi_{00}}, \\
           (p_{o} = ) \bra{\Phi_{00}} (\mathbb{I} \otimes \mathcal{N}) (\rho_{10}) \ket{\Phi_{00}} \geq \bra{\Phi_{00}} (\mathbb{I} \otimes \mathcal{M}) (\rho_{10}) \ket{\Phi_{00}}, \\
            (p_{q} = )\bra{\Phi_{00}} (\mathbb{I} \otimes \mathcal{N}) (\rho_{11}) \ket{\Phi_{00}} \leq \bra{\Phi_{00}} (\mathbb{I} \otimes \mathcal{M}) (\rho_{11}) \ket{\Phi_{00}}.
   \end{split}
 \end{equation}
As long as each value is checked, the exact formation of the conditions used is irrelevant. The operational interpretation of these alternative conditions is not known however. This complete the proof of Result~\ref{result5}. 

\def\arraystretch{2}
\begin{table}
\begin{tabular}{c|c|c|c|c|}
\cline{2-5}
\multicolumn{1}{l|}{}                 & $\rho_{00}$ & $\rho_{01}$ & $\rho_{10}$ & $\rho_{11}$ \\ \hline
\multicolumn{1}{|c|}{$\choi{N^{\dagger}}_{00}$} & $p_{l}$     & $p_{m}$     & $p_{0}$     & $1-3p_{q}$  \\ \hline
\multicolumn{1}{|c|}{$\choi{N^{\dagger}}_{01}$} & $p_{l}$     & $p_{m}$     & $1-3p_{0}$  & $p_{q}$     \\ \hline
\multicolumn{1}{|c|}{$\choi{N^{\dagger}}_{10}$} & $p_{l}$     & $1-3p_{m}$  & $p_{0}$     & $p_{q}$     \\ \hline
\multicolumn{1}{|c|}{$\choi{N^{\dagger}}_{11}$} & $1-3p_{l}$  & $p_{m}$     & $p_{0}$     & $p_{q}$     \\ \hline
\end{tabular}
\caption{The left most column, $X$, gives different elements of a POVM produced from the Choi-state of a channel $\mathcal{N}^{\dagger}$ using different Bell basis states in the definition. The top most row, $Y$, gives the four different states generated from the vectors that are also used to generate the four monotones. The values in the table are $\textrm{Tr}(XY)$, the expectation of value of getting an outcomes $X$ if the state is $Y$.  \label{Table3:operationalStratriges}}
\end{table}

\section{\label{Appendix six}Supplementary Material F: The Holevo Capacity is a Monotone}
The Holevo capacity is given by 
\begin{equation}
    \Holevo(\mathcal{N}) = \sup_{\{p_{i}, \rho_{i}\}} \big\{ s\big(\sum_{i} p_{i} \mathcal{N}(\rho_{i}) \big) - \sum_{i} p_{i} s\big(\mathcal{N}(\rho_{i} \big) \big\}, \label{HolevoCapacity}
\end{equation}
where $s(\cdot)$ is the von-Neumann entropy, and the supremum is taken over the ensemble of quantum states $\rho_{i}$ with probability $p_{i}$ \cite{amosov2000additivity, PhysRevA.69.022302, siudzinska_2020}. We present the following result:

\begin{lemma}
    The Holevo capacity is a monotone of both $\dynamicRT$ and $\dynamicRT_{n}$. Namely, $\Holevo(\mathcal{N})=0$ if $\mathcal{N} = \Lambda$, and $\Holevo(\mathcal{M}) \leq \Holevo(\mathcal{N})$ if there exists a $\Pi \in \freeD$ such that $\mathcal{M}=\Pi(\mathcal{N})$, where $\mathcal{M,N} \in \freeS$.
\end{lemma}

\begin{proof}
Firstly, it is noted that the Holevo capacity is invariant under pre-processing and post-processing unitary channels. Consider the channel $\Tilde{\mathcal{N}}(\rho) = V \mathcal{N}(U \rho U^{\dagger})V^{\dagger}$, where $U, V$ are unitary operators. The Holevo capacity of $\Tilde{\mathcal{N}}$ is  
\begin{equation}
    \begin{split}
        \Holevo(\Tilde{\mathcal{N}}) = \sup_{\{p_{i}, \rho_{i}\}} \big\{ s\big(\sum_{i} p_{i} V\mathcal{N}(U\rho_{i}U^{\dagger})V^{\dagger} \big) - \sum_{i} p_{i} s\big(V\mathcal{N}(U\rho_{i}U^{\dagger})V^{\dagger} \big) \big\}.
    \end{split}
\end{equation}
By defining $\eta_{i} = U\rho_{i}U^{\dagger}$ and $\rho' = \sum_{i} p_{i} \eta_{i}$ the above becomes 
\begin{equation}
    \begin{split}
        \Holevo(\Tilde{\mathcal{N}}) &= \sup_{\{p_{i}, \eta_{i}\}} \big\{ s\big(V\mathcal{N}(\rho')V^{\dagger} \big) - \sum_{i} p_{i} s\big(V\mathcal{N}(\eta_{i})V^{\dagger} \big) \big\}, \\
         &= \sup_{\{p_{i}, \eta_{i}\}} \big\{ s\big(\sum_{i} p_{i} \mathcal{N}(\eta_{i})\big) - \sum_{i} p_{i} s\big(\mathcal{N}(\eta_{i})\big) \big\}, \\
        &= \Holevo(\mathcal{N}).
    \end{split}
\end{equation}
Recall that the ability to diagonalise the Choi-states of qubit unital channels physically means that they are equivalent to a Weyl-covariant channel up to some pre-processing and post-processing unitary channels. Therefore, the above results show that the Holevo capacity of a qubit unital channel can equally be assessed through the Weyl covaraint channel that it is equivalent to - up to some pre-processing and post-processing unitary channels. Hence, only the Holevo capacity of the equivalent Weyl-covariant channels needs to be considered. Within the remainder of the proof all channels are considered to be Weyl-covariant. This ensures the Holevo capacity if a monotone of both $\dynamicRT$ and $\dynamicRT_{n}$.

For irreducibly covariant channels, of which Weyl-covariant channels are a subset, the Holevo capacity of a channel is simplified to be linearly proportional to the minimum output entropy \cite{amosov2000additivity, PhysRevA.69.022302, holevo2002remarks}. The Holevo capacity of a channel $\mathcal{N} \in \mathfrak{U}_{\mathcal{W}}$ is given by 
\begin{equation}
    \Holevo(\mathcal{N}) = \textrm{ln}(d_{S}) - \min_{\rho}s(\mathcal{N}(\rho)), 
\end{equation}
where the minimisation is over all states $\rho$. This can be rewritten as 
\begin{equation}
    \Holevo(\mathcal{N}) = \max_{\rho} D(\mathcal{N}(\rho) \vert \vert \Upsilon),
\end{equation}
where $\Upsilon$ is the maximally mixed state, and $D(\rho  \vert \vert \sigma) := \textrm{Tr}(\rho(\textrm{log}_{2}\rho - \textrm{log}_{2}\sigma))$ is the \textit{quantum relative entropy}. 

Consider now a channel $\mathcal{M}$, where there exists an allowed operations $\Pi \in \freeD$, such that $\mathcal{M} = \Pi(\mathcal{N})$, where $\mathcal{M,N} \in \freeS$. The Holevo capacity of $\mathcal{M}$ is given by 
\begin{equation}
    \begin{split}
        \Holevo(\mathcal{M}) &= \max_{\rho} D(\mathcal{M}(\rho) \vert \vert \Upsilon), \\
        &= \max_{\rho} D \Bigg( \sum_{\kappa} p_{\kappa} \mathcal{E}_{\kappa} \circ \mathcal{N} \circ \mathcal{P}_{\kappa} (\rho) \vert \vert \Upsilon \Bigg), \\
        &\leq \sum_{\kappa} p_{\kappa} \max_{\rho} D( \mathcal{E}_{\kappa} \circ \mathcal{N} \circ \mathcal{P}_{\kappa} (\rho) \vert \vert \Upsilon ).
    \end{split}
\end{equation}
Weyl-covariant channels commute with Weyl-covariant channels. Hence, a channel $\mathcal{Q}_{\kappa} = \mathcal{E}_{\kappa} \circ \mathcal{P}_{\kappa}$ can be defined such that 
\begin{equation}
    \begin{split}
        \Holevo(\mathcal{M}) &\leq \sum_{\kappa} p_{\kappa} \max_{\rho} D( \mathcal{Q}_{\kappa} \circ \mathcal{N} (\rho) \vert \vert \Upsilon ), \\
        &= \sum_{\kappa} p_{\kappa} \max_{\rho} D( \mathcal{Q}_{\kappa} \circ \mathcal{N} (\rho) \vert \vert \mathcal{Q}_{\kappa} (\Upsilon)), \\
        &\leq \sum_{\kappa} p_{\kappa} \max_{\rho} D( \mathcal{N} (\rho) \vert \vert \Upsilon) \\
        &= \max_{\rho} D( \mathcal{N} (\rho) \vert \vert \Upsilon),
    \end{split}
\end{equation}
where $\mathcal{Q}_{\kappa}(\Upsilon)=\Upsilon$ as $\mathcal{Q}_{\kappa}$ is a unital channel, and in the penultimate line the data processing inequality has been used such that
\begin{equation}
    D( \mathcal{Q}_{\kappa} \circ \mathcal{N} (\rho) \vert \vert \mathcal{Q}_{\kappa} (\Upsilon)) \leq D( \mathcal{N} (\rho) \vert \vert \Upsilon)~\forall~\kappa.
\end{equation}
Therefore, 
\begin{equation}
    \Holevo(\mathcal{M}) \leq \Holevo(\mathcal{N}),
\end{equation}
if $\mathcal{M}=\Pi(\mathcal{N})$. Hence, the Holevo capacity is a monotone of $\dynamicRT$ and $\dynamicRT_{n}$. 
\end{proof}
\section{\label{Appendix seven}Supplementary Material G: All monotones of $\staticRT$ are monotones of $\dynamicRT$.}

When considering the Choi-states of channels in $\freeS$, all monotones of $\staticRT$ are also monotones of $\dynamicRT$. We formalise this in the following Lemma. 
\begin{lemma} \label{lemma7.1}
\textit{If there exists a function $M(\cdot)$ such that}
\begin{equation}
    M(\rho) \leq M(\sigma) ~\textrm{if}~\rho=\mathcal{T}(\sigma), ~\mathcal{T}~\in~\freeS,
\end{equation}
then
\begin{equation}
    M(\choi{M}) \leq M(\choi{N})~~\textrm{if}~\mathcal{M}=\Pi(\mathcal{N}), ~\Pi~\in~\freeD, ~\mathcal{M,N}~\in~\freeS.
\end{equation}
\end{lemma}
\begin{proof}
The function $M(\cdot)$ can be defined to be a monotone of $\staticRT$ as it only decreases under allowed operations, $\freeS$. 

In the Choi-state formalism of $\dynamicRT$, allowed operations in $\freeD$ are of the form 
\begin{equation}
    \choi{M} = \sum_{\kappa} p_{\kappa} (\mathcal{P}_{\kappa} \otimes \mathcal{E}_{\kappa}) \big( \choi{N} \big),
\end{equation}
where $\choi{M}, \choi{N} \in \mathcal{H}_{a} \otimes \mathcal{H}_{b}$ and $\mathcal{P}_{\kappa}, \mathcal{E}_{\kappa} \in \freeS ~ \forall ~ \kappa$. Rewriting the channels $\mathcal{P}_{\kappa}$ and $ \mathcal{E}_{\kappa}$ in their noisy operations form gives 
\begin{equation}
    \begin{split}
        \mathcal{P}_{\kappa}(\cdot) & = \mathrm{tr}_{c}\Big(U^{\kappa}_{ac}((\cdot)_{a} \otimes 
    \frac{1}{d_{c}} \mathbb{I}_{c})(U^{\kappa}_{ac})^{\dagger}\Big), \\
        \mathcal{E}_{\kappa}(\cdot) & = \mathrm{tr}_{d}\Big(V^{\kappa}_{bd}((\cdot)_{b} \otimes 
    \frac{1}{d_{d}} \mathbb{I}_{d})(V_{bd}^{\kappa})^{\dagger} \Big),
    \end{split}
\end{equation}
where $\mathcal{H}_{c}$ and $\mathcal{H}_{d}$ are ancillary spaces in which a maximally mixed state is appended in the noisy operation. For a given $\kappa$ the operation $\mathcal{P}_\kappa \otimes \mathcal{E}_{\kappa}$ can be rewritten as 
\begin{equation}
    \begin{split}
        \mathcal{P}_{\kappa} \otimes \mathcal{E}_{\kappa} &= \mathrm{tr}_{c}\Big(U^{\kappa}_{ac}((\cdot)_{a} \otimes 
        \frac{1}{d_{c}} \mathbb{I}_{c})(U^{\kappa}_{ac})^{\dagger}\Big) \otimes  \mathrm{tr}_{d}\Big(V^{\kappa}_{bd}((\cdot)_{b} \otimes 
        \frac{1}{d_{d}} \mathbb{I}_{d})(V^{\kappa}_{bd})^{\dagger} \Big) \\
        &= \mathrm{tr}_{cd} \Big( U^{\kappa}_{ac}((\cdot)_{a} \otimes 
        \frac{1}{d_{c}} \mathbb{I}_{c})(U^{\kappa}_{ac})^{\dagger} \otimes V^{\kappa}_{bd}((\cdot)_{b} \otimes 
        \frac{1}{d_{d}} \mathbb{I}_{d})(V^{\kappa}_{bd})^{\dagger} \Big) \\
        &= \mathrm{tr}_{cd} \Big( (U^{\kappa}_{ac} \otimes V^{\kappa}_{bd}) \big( (\cdot)_{ab} \otimes \frac{1}{d_{c}d_{d}} \mathbb{I}_{cd} \big) (U^{\kappa}_{ac} \otimes V^{\kappa}_{bd})^{\dagger} \Big) \\
        &= \mathrm{tr}_{\alpha} \Big( Q^{\kappa}_{abcd} \big( (\cdot)_{ab} \otimes \frac{1}{d_{\alpha}} \mathbb{I}_{\alpha} \big) (Q^{\kappa}_{abcd})^{\dagger} \Big).
    \end{split}
\end{equation}
where $\alpha=cd$. An allowed operation from $\freeD$ in the Choi-state formalism is a noisy operation on the state $\choi{N}$. This is an allowed operations of $\staticRT$ on the biparitie space of $\mathcal{H}_{a} \otimes \mathcal{H}_{b}$. The static resource content of the Choi-state of a channel in $\freeS$ can therefore only decrease under allowed operations in $\freeD$. This leads to all monotones of $\staticRT$ also being monotones of $\dynamicRT$ when considering the Choi-state formalism. 
\end{proof}
\end{document}